\documentclass[10pt,a4paper,oneside]{article}
\usepackage[utf8]{inputenc}
\usepackage{calc}
\usepackage[english]{babel}
\usepackage{mathtools, amsmath, amssymb, amsthm, bbm}
\usepackage{physics}
\usepackage{color}
\usepackage[normalem]{ulem}
\usepackage[backend=biber, style=numeric, sorting=none]{biblatex}
\usepackage{hyperref}

\usepackage{graphicx}
\usepackage{caption}
\usepackage{subcaption}
\usepackage{authblk} 

\newtheorem{theorem}{Theorem}[section]
\newtheorem{lemma}[theorem]{Lemma}
\newtheorem{proposition}[theorem]{Proposition}
\newtheorem{corollary}[theorem]{Corollary}

\newtheorem{remark}[theorem]{Remark}
\newtheorem{example}[theorem]{Example}
\newtheorem{assumption}[theorem]{Assumption}

\newtheorem{result}[theorem]{Result}

\def\res{\mathop{Res}\limits}
\newcommand{\bZ}{\mathbb Z}

\numberwithin{equation}{section}

\title{Scaling Properties of Current Fluctuations in Periodic TASEP}

\author[1]{Anastasiia Trofimova\thanks{anastasiia.trofimova@gssi.it}}
\author[1]{Lu Xu \thanks{lu.xu@gssi.it}}

\affil[1]{Gran Sasso Science Institute, Viale Francesco Crispi 7, 67100 L’Aquila, Italy
}
\date{}
\begin{document}

\emergencystretch 3em
	
\maketitle
	
	\begin{abstract}
		We study current fluctuations in the Totally Asymmetric Simple Exclusion Process (TASEP) on a ring with $N$ sites and $p$ particles.
        By introducing a deformation parameter $\gamma$, we analyze the tilted operator that governs the statistics of the time-integrated current.
        Employing the coordinate Bethe ansatz, we derive implicit expressions for the scaled cumulant generating function (SCGF), i.e. the largest eigenvalue, and the spectral gap, both in terms of Bethe roots.
        Their asymptotic behaviour is characterized by using the geometric structure of Cassini oval.

        In the thermodynamic limit at fixed particle density, we identify a dynamical phase transition separating fluctuation regimes. For $\gamma>0$, the SCGF exhibits ballistic growth with system size, $\lambda_1 \sim N$.
        In contrast, for $\gamma<0$, the SCGF converges to $-1$ as $N\to\infty$.
        This transition is reflected in the spectral gap, which controls the system's relaxation timescale.
        For $\gamma>0$, the gap closes at polynomial speed, $\Delta \sim N^{-1}$, consistent with rapid relaxation with enhanced current.
        For $\gamma<0$, the gap vanishes exponentially, $\Delta \sim \exp(-cN)$, signaling metastability with diminished current.
        Our non-perturbative results provide insights into large deviations and the relaxation dynamics in driven particle systems.
	\end{abstract}
	\tableofcontents
	\newpage
	\section{Introduction}  
The Totally Asymmetric Simple Exclusion Process (TASEP) is a fundamental model in non-equilibrium statistical mechanics that captures essential features of driven transport in one-dimensional systems \cite{liggett1985interacting, Spohn91}.
In TASEP, particles hop unidirectionally along a discrete lattice with hard-core repulsion: each site can hold at most one particle.
As an exactly solvable model in the Kardar--Parisi--Zhang (KPZ) universality class \cite{KardarParisiZhang1986dynamic}, it serves as a prototype for a broad variety of stochastic growth phenomena.

Under periodic boundary conditions, the Bethe ansatz provides exact results for the spectrum of the Markov generator.
The spectral gap, which corresponds to the inverse of the longest relaxation time of the dynamics, scales as $N^{-\frac32}$ with system size $N$. This was proved at half-filling in the seminal work \cite{GwaSpohn1992spgap} and for arbitrary density in \cite{kim1995bethe}, see also \cite{2004GolinelliMallick, GolinelliMallick2005spectral}.

Furthermore, the Bethe ansatz enables the analysis of the large deviations of the \emph{time-integrated particle current} \cite{1999LebowitzGallavotti, derrida1998exactLDfunction, 1999Derrida1Appert}.
By tilting the generator with a factor $e^\gamma$ in the hopping rate, one biases the dynamics toward atypical currents.
The largest eigenvalue of the tilted operator gives the scaled cumulant generating function (SCGF) of the current.
Its Legendre--Fenchel transform then yields the large deviation rate function.
Particularly in \cite{1999Derrida1Appert}, the authors computed the SCGF for all $\gamma$ in closed form.
They showed that for $\gamma>0$ the SCGF grows linearly in $N$, and for $\gamma<0$ it converges to $-1$ as $N\to\infty$.
They also characterized the distinct scaling behaviours of the left and right tails of the large deviation rate function, which is conjectured to be universal within the KPZ class.

The research landscape has expanded to the partially asymmetric case, in which particles jump to the right at rate $p$ and left at rate $q\not=p$.
Conditioning on atypically large current, space-time correlations, conformal invariance, the emergence of hyperuniform macrostates, and the different optimal transport mechanisms are explored by studies in the limit $\gamma \to +\infty$ \cite{1995Schutz, 2017KarevskiSchutzConformal, 2010PopkovSchutzSimon}.

At the relaxation timescale $t \sim N^{\frac32}$, the current fluctuation in periodic TASEP exhibits a crossover between KPZ-type scale and the usual Gaussian scale, as observed in \cite{prolhac2016finite} at half-filling and in \cite{BaikLiu2018fluctuations, liu2018height} for general densities.
Recent works \cite{BaikLiu2019multipoint, BaikLiu2021periodic} have provided exact space-time joint distributions using the determinant representation of the transient probabilities at any time \cite{schutz1997exact, priezzhev2003exact}.
For TASEP on the infinite lattice, the fluctuation size grows as $t^{\frac13}$ and the scaling limit is given by the Tracy--Widom distribution.
This is first proved in \cite{johansson2000shape} by mapping the dynamics to last passage percolation, see also \cite{BaikRains2000limiting, PrahoferSpohn2002current, RakosSchutz2005current}.
The study is then extended to the dynamics of tagged particle \cite{ImamuraSasamoto2007dynamics}, as well as the joint distributions at multiple space-time points and along space-like paths \cite{BorodinFerrari2007fluctuation, BorodinFerrari2008large, BorodinFerrariSasamoto2008large, PovolotskyPriezhevSchutz2011generalized, PoghosyanPovolotskyPriezzhev2012universal}.

Complementing these probabilistic and spectral methods, recent study on vertex models related to TASEP has deepened understanding of the macroscopic structure. In \cite{deGierKenyonWatson2021}, the authors considered the asymmetric five-vertex model---a relative of the six-vertex model which generalizes TASEP configurations. They computed its free energy defined via the largest eigenvalue of the transfer matrix, and the associated surface tension function, providing a variational principle for limit shapes. This geometric approach to analyzing integral equations that arise from Bethe ansatz-like structures helps further clarify the connection between microscopic configurations and macroscopic profiles.

Recent advances have developed a geometric approach to the Bethe ansatz, based on the analysis of the complex structure of Bethe roots using Riemann surfaces and Cassini ovals \cite{prolhac2020riemann, prolhac2022riemann}.  This method was motivated by the challenge of describing the full spectrum of the tilted operator. Employed to TASEP, \cite{2025IwaoMotegiCompleteness} showed the completeness of the Bethe ansatz, i.e. the Bethe ansatz equations yield all eigenstates, by counting connected components and classifying spectral degeneracies. It complements the earlier result proved for partially asymmetric exclusion with a generic set of hopping rates \cite{2017BrattainDoSaenz}.

The goal of the present paper is to investigate the particle current in periodic TASEP when the deformation parameter $\gamma$ is finite.
For all $\gamma \in \mathbb R$, we analyze the SCGF and the spectral gap of the tilted operator in the thermodynamic limit $N\to\infty$ with the density $\rho$ fixed, and perform exact results of the leading order as functions of $(\gamma,\rho)$.
Roughly speaking, under specific conjectures on the structure of the Bethe roots, we show the following.
\begin{itemize}
\item The SCGF grows linearly in $N$ for $\gamma>0$: $\lambda_1 \sim N\Lambda(\gamma,\rho)$, while for $\gamma<0$, $\lambda_1 \sim -1+e^{-cN}$, see Result \ref{result: lambda1}.
\item The spectral gap scales inversely with $N$ for $\gamma>0$: $\Delta \sim N^{-1}g(\gamma,\rho)$, while for $\gamma<0$, $\Delta$ vanishes exponentially, see Result \ref{result: sp gap}.
\end{itemize}
The precise definition of the functions $\Lambda$, $g$ is given in Section \ref{sec: asym+}, Theorem \ref{thm: asym+}.
The result for $\gamma>0$ enlightens the transition between the KPZ fluctuation ($\gamma=0$) and the ballistic fluctuation ($\gamma\to+\infty$).
The exponential decay of the spectral gap for $\gamma<0$ reflects the metastability of the system with suppressed current.

The paper is organized as follows. Section \ref{sec: Model and results} defines the model and summarizes the main results. Section \ref{sec: sp properties} is devoted to the analysis of the spectral properties using the Bethe ansatz. We also state rigorously in Section \ref{subsec: ass} our assumptions on the correspondence between Bethe roots and the largest and the second largest eigenvalues. The proof of the thermodynamic limit is given in Section \ref{sec: Asymptotic analysis}. Finally, an a priori bound and detailed computation concerning the Cassini oval are given in Appendix \ref{app boundedness} and \ref{app: contour integration}, respectively.

\section{Model and main results} \label{sec: Model and results}
    \subsection{State space and dynamics of TASEP}
	TASEP is a one-dimensional stochastic interacting particle model formulated as a continuous time Markov process $\boldsymbol{\eta}(t): \mathbb{R}_{\geq 0} \rightarrow \Omega_{N,p}$ on the state space 
	$$
	\Omega_{N,p} := \biggl\{ \boldsymbol{\eta} = (\eta_1, \dots, \eta_N) \in \{0,1\}^{\bZ_N} \;\Big|\; \sum_{x=1}^N \eta_x = p \biggr\},
	$$
consisting of particle configurations on a periodic one-dimensional lattice $\bZ_N = \bZ/N \bZ$ with $N$ sites (sites $i$ and $i+N$ are identified) and $p$ particles, with each site occupied by at most one particle.
The number of configurations is
     \begin{equation} \label{eq: def of |Omega|}
	 |\Omega_{N,p}| = \binom{N}{p}.
	 \end{equation}
	
	The dynamics evolves in continuous time as follows. Each particle has an independent exponential clock of rate one. When a particle's clock rings, it jumps to the neighbouring right site only if it is empty. 
     The dynamics preserves the total number of particles and hence the particle density $\rho := p/N$. 

     Fix an initial probability distribution $P_0$ on $\Omega_{N,p}$. Let $P_t$ be the probability distribution at time $t$: $P_t(\boldsymbol{\eta}) := P(\boldsymbol{\eta}(t)=\boldsymbol{\eta})$. Its evolution is governed by the master equation (Kolmogorov forward equation)
	 	\begin{equation}
	 		\label{eq: Master eqn}
	 		\frac{d}{dt} P_t(\boldsymbol{\eta}) = (L P_t)(\boldsymbol{\eta}),
	 	\end{equation}
	 	where generator $L$ is a linear operator defined by
	 	\begin{equation} \label{eq: def L operator}
	 		L P_t(\boldsymbol{\eta}) = \sum_{x=1}^N \eta_{x+1}(1-\eta_{x}) \left( P_t({\boldsymbol{\eta}}^{x,x+1}) - P_t(\boldsymbol{\eta}) \right).
	 	\end{equation}
	 Here, $\boldsymbol{\eta}^{x,x+1}$ denotes the configuration obtained from $\boldsymbol{\eta}$ by exchanging the occupation numbers $\eta_{x}$ and $\eta_{x+1}$.
     
     The action of $L$ has a matrix representation: $LP_t=MP_t$, where $M$ is a transition rate matrix of size $|\Omega_{N,p}|$ such that
     	\begin{equation} \label{eq: def M}
		M(\boldsymbol{\eta}, \boldsymbol{\eta}') := 
\begin{cases}
    \eta_{x+1}(1 - \eta_x), & \text{if } \boldsymbol{\eta}' = \boldsymbol{\eta}^{x,x+1}, \\
    -\sum_{x=1}^N \eta_{x+1}(1 - \eta_x), & \text{if } \boldsymbol{\eta}' = \boldsymbol{\eta}.
\end{cases}
\end{equation}
As the generator of a finite irreducible Markov process, the largest eigenvalue of $M$ is $0$, and the unique stationary distribution is uniform over $\Omega_{N,p}$
 \begin{equation} \label{eq: uniform distribution}
	 	P^{\mathrm{stat}}(\boldsymbol{\eta}) = |\Omega_{N,p}|^{-1}, \qquad  \forall\,\  \boldsymbol{\eta} \in \Omega_{N,p}.
	 \end{equation}

	\subsection{Particle current}
	Consider the \emph{time-integrated total particle current} $Y(t): \mathbb{R}_{\geq 0} \rightarrow \mathbb{N}$, which is the additive random variable on the trajectories of the process that counts the total number of rightward jumps that have occurred across all sites up to time $t$. In other words, $Y(0)=0$ and every time a particle moves to the right at any site, the value of $Y(t)$ increases by $1$.
    
    To study the evolution of $Y(t)$, observe that the joint process $(\boldsymbol{\eta}(t), Y(t))$ is a Markov process on the state space $\Omega_{N,p} \times \mathbb{N}$.
    Denote by $p_t(\boldsymbol{\eta}, Y)$ the joint probability distribution of $(\boldsymbol{\eta}(t),Y(t))$:
    $$
    p_t(\boldsymbol{\eta}, Y) := P\big(\boldsymbol{\eta}(t)=\boldsymbol{\eta},Y(t)=Y\big), \qquad (\boldsymbol{\eta},Y) \in \Omega_{N,p} \times \mathbb{N}.
    $$
    The master equation of $p_t$ reads
    \begin{equation} \label{eq: joint Markov operator}
        \frac{d}{dt}p_t(\eta,Y) = \sum_{x=1}^N \eta_{x+1}(1-\eta_{x}) \left( p_t({\boldsymbol{\eta}}^{x,x+1},Y-1) - p_t(\boldsymbol{\eta},Y) \right).
    \end{equation}
    %
    %
    
    Restricted to any given final state $\boldsymbol{\eta}(t) = \boldsymbol{\eta}$, define the moment-generating function of $Y(t)$ by
	\begin{equation}
		G_{t,\gamma}(\boldsymbol{\eta}) := \sum_{Y = 0}^{\infty} p_t(\boldsymbol{\eta}, Y) e^{\gamma Y}, \qquad \gamma \in \mathbb R.
	\end{equation}
    Using \eqref{eq: joint Markov operator} and $Y(0) \equiv 0$, $G_{t,\gamma}$ solves the following initial value problem:
	\begin{equation}\label{eq: evolution of mg function}
		\frac{d}{dt} G_{t,\gamma}(\boldsymbol{\eta}) = L_\gamma  G_{t,\gamma}(\boldsymbol{\eta}), \qquad G_{0,\gamma} = P_0,
	\end{equation}
	where $P_0$ is the initial distribution of $\boldsymbol{\eta}$, and $L_{\gamma}$ is the continuously deformed version of the operator $L$ given by (cf. \eqref{eq: def L operator})
	\begin{equation} \label{eq: def of L_gamma}
		L_\gamma g(\boldsymbol{\eta}) := \sum_{x=1}^N \eta_{x+1} (1-\eta_{x}) \left( e^\gamma g(\boldsymbol{\eta}^{x,x+1}) - g(\boldsymbol{\eta}) \right).
	\end{equation}
Here, the parameter $e^{\gamma}$ weights each jump, counting the number of jumps in a generating function sense.

Similarly to \eqref{eq: def M}, let $M_\gamma$ be the matrix representation of $L_{\gamma}$. It can be obtained from $M$ by multiplying the non-diagonal elements with $e^{\gamma}$ and keeping the diagonal ones. Precisely,
	\begin{equation} \label{eq:def Mgamma}
		M_{\gamma}(\boldsymbol{\eta}, \boldsymbol{\eta}') = 
        \begin{cases}
            e^\gamma \eta_{x+1}(1-\eta_{x}), & \text{if } \boldsymbol{\eta}' = \boldsymbol{\eta}^{x,x+1}, \\
            -\sum_{x=1}^N \eta_{x+1}(1-\eta_{x}), & \text{if }  \boldsymbol{\eta}'= \boldsymbol{\eta}.
        \end{cases}
	\end{equation}
The solution to \eqref{eq: evolution of mg function} is explicitly given by
	\begin{equation} \label{eq: solution to G_t}
    G_{t,\gamma}(\boldsymbol{\eta})  = \left( e^{t M_{\gamma}}  P_0 \right) (\boldsymbol{\eta}).
	\end{equation}
	Summing up all the final states, we obtain the full moment-generating function of $Y(t)$
    \begin{equation}
    \mathbb{E} [e^{\gamma Y(t)}] = \left( e^{t M_{\gamma}} P_0, \mathbf{1} \right),
    \end{equation}
where $(f,g) := \sum_{\boldsymbol{\eta} \in \Omega_{N,p}} f(\boldsymbol{\eta})g(\boldsymbol{\eta})$ denotes the scalar product on $\Omega_{N,p}$, and by $\mathbf{1}$ we denote the vector with all components equal to $1$.

\subsection{Main spectral results and asymptotics} \label{sec: results}

The large-time behaviour of the moment-generating function $G_{t,\gamma}(\boldsymbol{\eta})$, which encodes the statistics of the time-integrated particle current $Y(t)$, is governed by the spectral properties of the tilted matrix $M_{\gamma}$. Namely, its eigenvalues generally control the asymptotic growth rates of quantities related to the current.

More precisely, 
for $j=1$, ..., $|\Omega_{N,p}|$, denote by $\lambda_j(\gamma)$ the eigenvalues of $M_\gamma$, in which $\lambda_1(\gamma)$ has the largest real part.
By Perron--Frobenius theorem, $\lambda_1(\gamma)$ is real and simple.
This eigenvalue determines the exponential growth rate of the moment-generating function via the limit
\[
\lambda_1(\gamma) = \lim_{t \to \infty} \frac{1}{t} \log \mathbb{E}[e^{\gamma Y(t)}].
\]
Thus, the large deviations of $Y(t)$ reduces to $\lambda_1(\gamma)$.
Moreover, let $\lambda_2(\gamma)$ be one of the eigenvalues with the second-largest real part.
Then, $\Re{\lambda_1(\gamma) - \lambda_2(\gamma)}$, the spectral gap, governs the rate of convergence in the limit above, characterizing transient relaxation effects.

Our first result includes some basic properties of $\lambda_1(\gamma)$ as a function of $\gamma$ and gives an a priori estimate.

 \begin{result} \label{result: bounds}
 For any $N$ and $p$, the largest eigenvalue $\lambda_{1}(\gamma)$ is a continuous, strictly increasing function of $\gamma$ and satisfies the following bounds
 \begin{align*}
     e^{\gamma}-1 \leq \lambda_{1}(\gamma) \leq p (e^{\gamma}-1), \qquad \forall \gamma >0,\\
     p(e^{\gamma}-1) \leq \lambda_{1}(\gamma) \leq e^{\gamma}-1 , \qquad \forall \gamma <0.
 \end{align*}
\end{result}

 To further investigate the moment-generating function in long-time, 
 a natural approach is through a spectral decomposition of the tilted generator $M_{\gamma}$. 
 Assuming that $M_\gamma$ is diagonalizable, the full set of eigenvalues allows us to express the moment-generating function as
 \begin{equation} \label{eq: spectral decomposition}
 \mathbb{E} [e^{\gamma Y(t)}] =  \left( e^{t M_{\gamma}} P_0, \mathbf{1} \right)= \sum_{j=1}^{|\Omega_{N,p}|} e^{t \lambda_j } \ a_j(N, p,  P_0).
 \end{equation}
 The coefficients $a_j$ in this decomposition do not depend on time and can be associated with the overlaps with initial distribution $P_0$.
 Isolating the first two terms of the spectral decomposition reveals that the speed of convergence to the asymptotic exponential growth is determined by the spectral gap.
 \begin{equation}
 \frac{\log \mathbb E \left[e^{\gamma Y(t)}\right]} {t}  =  \lambda_1(\gamma) + O(e^{-t (\lambda_1(\gamma)-\lambda_2(\gamma))} t^{-1} ).
 \end{equation}

\begin{remark}
    A significant challenge that remains open is to rigorously prove that $M_\gamma$ is diagonalizable, or even stronger, diagonalizable with the Bethe ansatz introduced below in Section \ref{sec:calc-bethe-equat}. See also Remark \ref{rem: iwaomotegi}.
\end{remark}

In this paper, we compute the largest eigenvalue and the spectral gap in the \emph{thermodynamic limit} $N\to\infty$ with particle density $\rho= p/N$ fixed. 
A summary of the main results is presented below. For precise formulations and proofs, see Assumptions \ref{ass: gammapos} and \ref{ass: gammaneg}, as well as Theorem \ref{thm: asym+} and Proposition \ref{prop: CN(gamma) asympt serie}.

 
\begin{result} \label{result: lambda1}
In the thermodynamic limit, the largest eigenvalue experiences a phase transition from the inactive phase of $\gamma<0$ to the active phase of $\gamma > 0$:
\begin{equation}
 \lambda_1(\gamma) = \begin{cases}
 N\,\Lambda(\gamma, \rho) +O(1), & \gamma>0, \\ 
 -1 + (e^{N \gamma \rho} + e^{N\gamma (1-\rho)})\big[1 + o(1)\big], & \gamma < 0.
 \end{cases}
\end{equation}
where $\Lambda(\gamma, \rho) \in \mathbb{R}_{>0}$ is a model-dependent constant defined in \eqref{def LambdaCoef}.
\end{result}

\begin{result} \label{result: sp gap}
In the thermodynamic limit, the spectral gap experiences a phase transition from the metastable phase of $\gamma < 0$ to the fast-relaxing phase of $\gamma>0$:
 \begin{equation}
\lambda_1(\gamma) -\lambda_2(\gamma)  = \begin{cases}
 N^{-1} \big[g(\gamma, \rho) +o(1)\big], & \gamma>0, \\ 
 o(e^{\gamma N \rho}) + o(e^{\gamma N (1-\rho)}),& \gamma < 0.
 \end{cases}
\end{equation}
where $g(\gamma, \rho) \in \mathbb{C}$ is a model-dependent constant with positive real part defined in \eqref{def gCoef}.
\end{result}

\begin{remark}
    The the exact formulas for the leading order of $\lambda_1(\gamma)$ (see Theorem \ref{thm: asym+} and Proposition \ref{prop: CN(gamma) asympt serie}) agrees with Eq.\,(53) and the series expansion proceeding Eq.\,(57) and the series expansions proceeding Eq.\,(25) of \cite{1999Derrida1Appert}.
    Moreover, as a function of $\gamma$, for each fixed $\rho$ we have
    \begin{equation}
        \Lambda(\gamma, \rho) =
        \begin{cases}
        e^{\gamma} \pi^{-1} \sin(\pi \rho) - \rho(1-\rho) + o(1), & \gamma \rightarrow +\infty,\\
        \gamma  \rho(1-\rho) + o(\gamma), & \gamma \rightarrow 0+.
    \end{cases}
    \end{equation}
    This agrees with the asymptotics predicted 
    in Eq.\,(14) and (15) of \cite{derrida1998exactLDfunction}. 
    
    Similarly, for $g(\gamma,\rho)$ we have
    \begin{equation}
        \Re{g(\gamma, \rho)} = 
        \begin{cases}
        2e^{\gamma}\pi \sin(\pi \rho) + O(1), & \gamma \rightarrow +\infty,\\
        4 \pi^{\frac 43} 3^{-\frac23} \gamma^{\frac 13} \left(\rho(1-\rho)\right)^{\frac23} +O(\gamma^{\frac 23}), & \gamma \rightarrow 0+.
        \end{cases}
    \end{equation}
    In particular, the leading order for $\gamma \to +\infty$ agrees with Eq.\,(74) of \cite{2010PopkovSchutzSimon}.
\end{remark}

We would like to emphasize that in the active regime corresponding to $\gamma>0$, the moment-generating function exponentially biases the dynamics towards trajectories with large particle currents. As the weighting factor $e^{\gamma}$ amplifies the contribution of histories where many particles jump, the system is promoted to sustain a high current flow and exhibits much faster relaxation timescale $t \sim N$.

In contrast, for $\gamma < 0$, the system remains in the inactive (or metastable) phase, where trajectories with large currents are exponentially suppressed by the factor $e^{\gamma}$. This effectively biases the system towards configurations exhibiting low activity, such as jammed or blocked states. 
Hence, convergence towards steady state can take exponentially long time.

\section{Spectral Properties  and Bethe Ansatz
}  \label{sec: sp properties}
\subsection{Bounds and SCGF qualitative behaviour}
     In this part, we analyze the spectral properties of the matrix $M_{\gamma}$, proving Result \ref{result: bounds}. 
    
	A key feature of $M_{\gamma}$ is that it can be shifted into a non-negative matrix by adding a sufficiently large constant multiple of the identity. More precisely, for any $\kappa \geq p$, the shifted matrix
$
M_{\gamma} + \kappa Id
$
has non-negative entries. Importantly, this shifted matrix shares the same eigenvectors as $M_{\gamma}$, with all eigenvalues simply shifted by $\kappa$.

In particular, taking $\kappa = p$ leads to a non-negative, irreducible, and aperiodic matrix. By the Perron–Frobenius theorem, this matrix has a unique simple dominant eigenvalue and a strictly positive right eigenvector (unique up to scaling). Translating back, the original matrix $M_{\gamma}$ thus has a largest eigenvalue $\lambda_1(\gamma) > -p$, corresponding to the same positive eigenvector.

The entries of $M_{\gamma}$ depend analytically on the parameter $\gamma$, which ensures that its eigenvalues and eigenvectors vary analytically with $\gamma$. In particular, the largest eigenvalue $\lambda_1(\gamma)$ varies continuously with $\gamma$, remaining simple for all $\gamma$, and satisfies the normalization condition $\lambda_1(0) = 0$ when $M_{\gamma}$ reduces to the original generator $M$.

It is also easy to see from the Perron--Frobenius theorem that $M_{\gamma}$ has the monotonicity property: if $\gamma_1 > \gamma_2$, then for the matrices $M_{\gamma_1} \geq M_{\gamma_2}$ (where the inequality is entrywise), we have $\lambda_1(\gamma_1) \geq \lambda_1(\gamma_2)$.
   
    \begin{proof}[Proof of Result \ref{result: bounds}]
    Let $\phi$ be the principal right eigenvector of $M_{\gamma}$ such that
    \begin{equation} \label{eq: master on components}
			(M_{\gamma} \phi)(\boldsymbol{\eta}) =
			\sum_{\boldsymbol{\eta}': \boldsymbol{\eta}' \rightarrow \boldsymbol{\eta}} \left( e^{\gamma} \phi (\boldsymbol{\eta}')
			- \phi(\boldsymbol{\eta}) \right)
			= \lambda_{1}(\gamma) \phi(\boldsymbol{\eta}).
		\end{equation}

We analyze this equation to obtain bounds on $\lambda_1(\gamma)$ by considering extremal values of $\phi(\boldsymbol{\eta})$. 
Looking at the smallest component $\beta := \min_{\boldsymbol{\eta}} \phi(\boldsymbol{\eta})$, $\beta > 0 $ of an eigenvector $\phi$, we have
		\begin{equation} \label{eq: master for beta}
			\sum_{\boldsymbol{\eta}': \boldsymbol{\eta}' \rightarrow \boldsymbol{\eta}_{min}} \left( e^{\gamma} \phi (\boldsymbol{\eta}')
			- \beta \right)  = \lambda_{1}(\gamma) \beta.
		\end{equation}
For $\gamma>0$, every term in the sum is positive and the number of terms varies from $1$ to $p$, therefore, $\lambda_1(\gamma)\geq e^{\gamma }-1$, while for $\gamma<0$, every term in the sum is negative, therefore, $\lambda_1(\gamma) \leq e^{\gamma }-1$.

        Similarly, for the largest component $\alpha := \max_{\boldsymbol{\eta}} \phi(\boldsymbol{\eta})$, $ \alpha >0$, we obtain
		\begin{equation}\label{eq: master for alpha}
			\sum_{\boldsymbol{\eta}': \boldsymbol{\eta}' \rightarrow \boldsymbol{\eta}_{max}} \left( e^{\gamma} \phi (\boldsymbol{\eta}')
			- \alpha \right)  = \lambda_{1}(\gamma) \alpha.
		\end{equation}
		
        For $\gamma>0$, every term in the sum is positive, therefore, $\lambda_1(\gamma) \leq p( e^{\gamma }-1)$. For $\gamma<0$, every term in the sum is negative, therefore, $\lambda_1(\gamma) \geq p(e^{\gamma }-1)$.
        \end{proof}
        
	\begin{figure}[ht]
		\centering
			\includegraphics[width=1\linewidth]{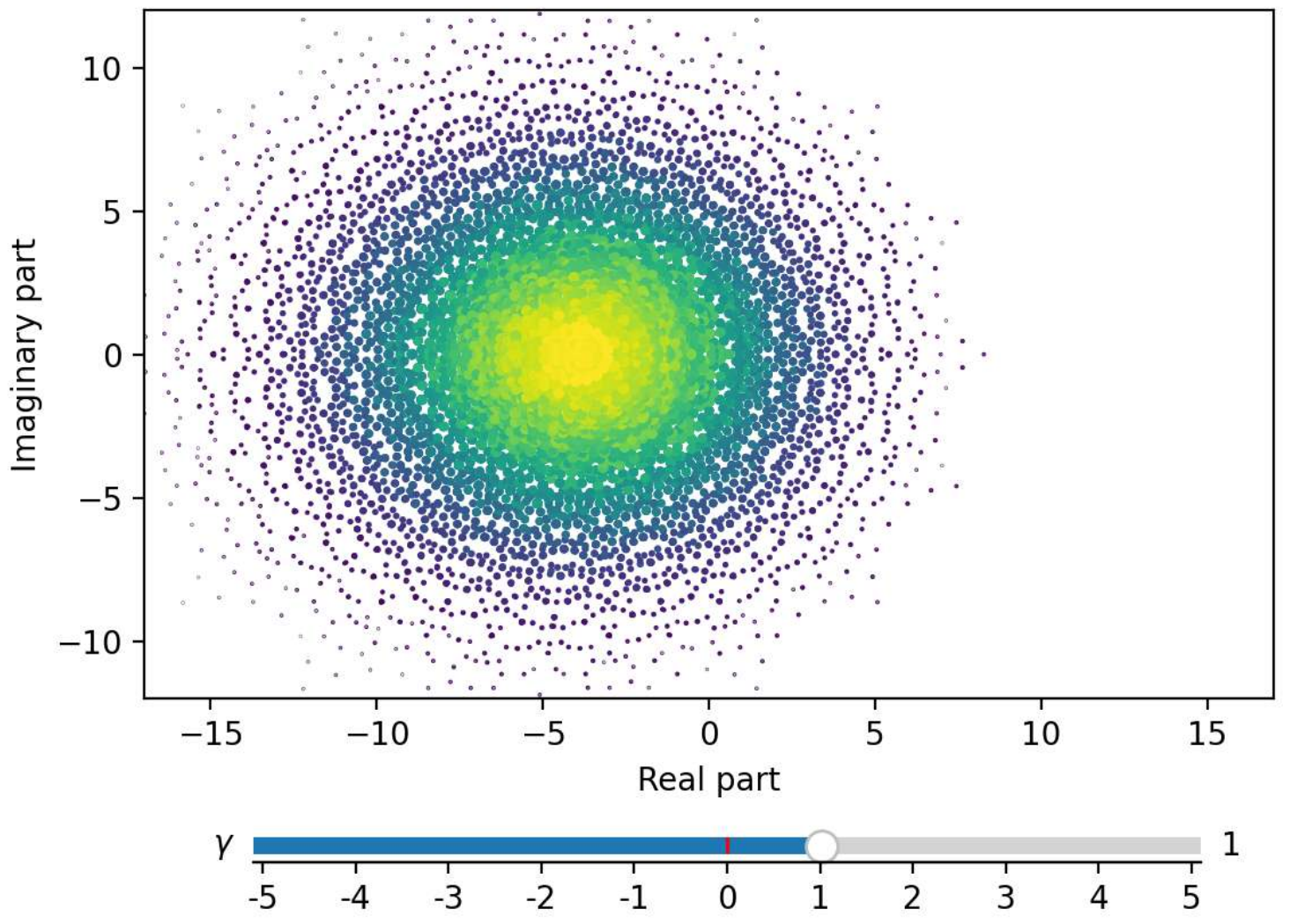}
		\caption{Eigenvalues of the operator $M_{\gamma}$ for positive $\gamma = 1$, $N= 18$, $p=6$ plotted on the complex plane. The colour reflects density: lighter colour indicates higher eigenvalue density.}
		\label{fig: eigvals pos gamma}
	\end{figure}
	Furthermore, the numerical calculations reveal that the eigenvalues  $M_{\gamma}$ exhibit qualitatively different patterns on the complex plane. For the positive values of parameter $\gamma$, the eigenvalues form a captivating rounded cloud with an intricate pattern. In contrast, for negative values of $\gamma$ (especially for quite large values of $N$), we observe several clusters of eigenvalues near the points $-1,-2,-3, \dots$, along with a compact tail extending along the negative real axis.
	This observation is illustrated in Figures \ref{fig: eigvals pos gamma} and \ref{fig: eigvals neg gamma}.
	
	\begin{figure}[ht]
		\centering
		\includegraphics[width=1\linewidth]{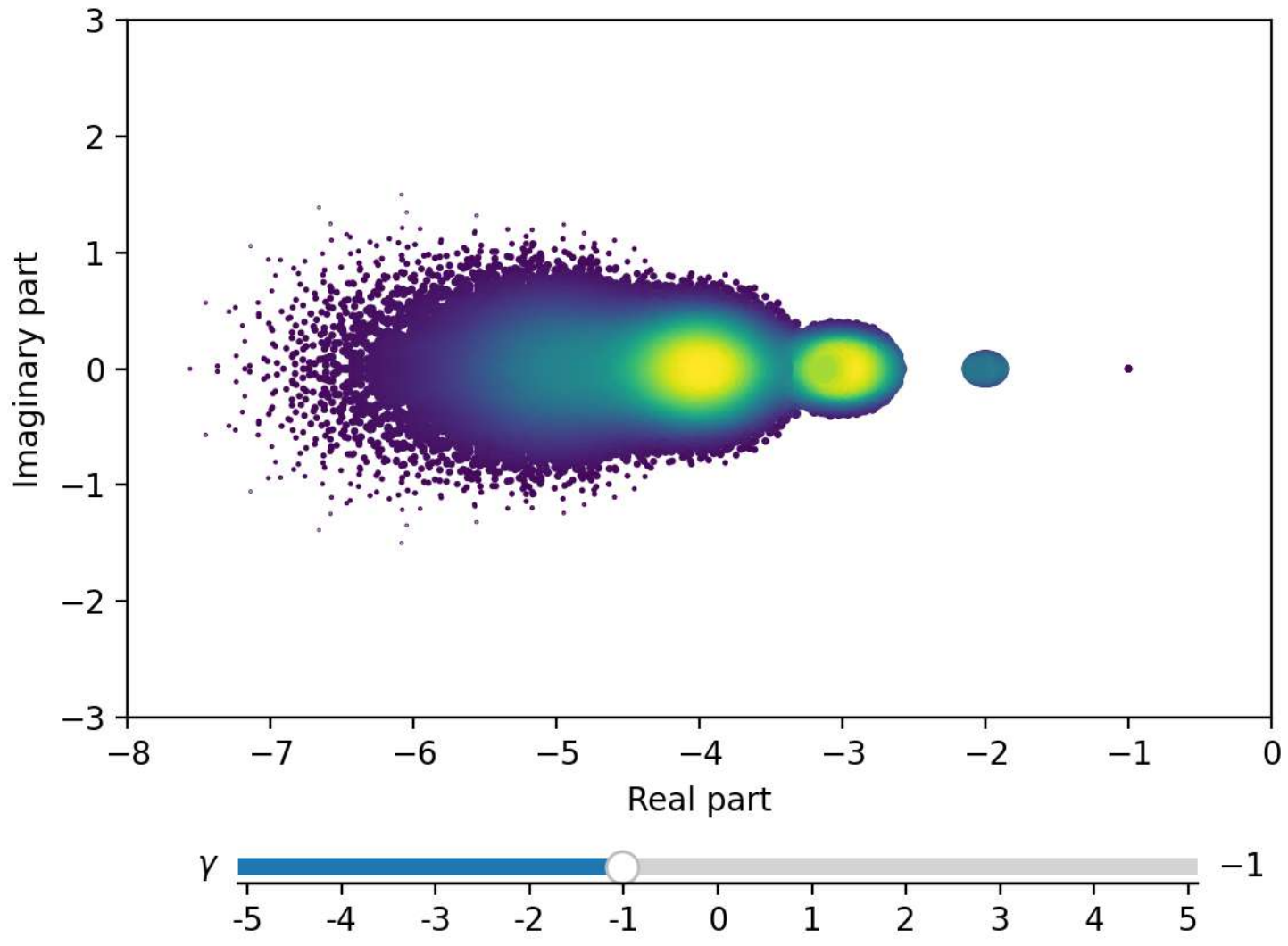}
		\caption{Eigenvalues of the operator $M_{\gamma}$ for negative $\gamma = -1$, $N= 18$, $p=6$ plotted on the complex plane. The colour reflects density: lighter colour indicates higher eigenvalue density.}
		\label{fig: eigvals neg gamma}
	\end{figure}

Additionally, the numerical analysis indicates that the behaviour of the spectral gap in the thermodynamic limit differs qualitatively depending on whether the parameter $\gamma$ is positive or negative. For positive values of $\gamma$,  as the system size grows, we find that both the largest eigenvalue and the next to the largest eigenvalue increase in real part. Conversely, for negative $\gamma$, as we approach the thermodynamic limit, both eigenvalues converge towards the value of $-1$.

\subsection{Bethe Ansatz } \label{sec:calc-bethe-equat}
	
    In this section, we fix $N$, $p\in\mathbb{N}$ and some $\gamma\in\mathbb{R}$ and study the eigenvalues of $M_{\gamma}$ using the \emph{coordinate Bethe ansatz}.
    
    Let $\lambda(\gamma) = \lambda(\gamma;N)$ be an eigenvalue of $M_\gamma$ and $\phi$ be a corresponding right eigenvector.
    Any configuration $\boldsymbol{\eta} \in \Omega_{N,p}$ can be equivalently described by the ordered particle positions $(x_1,x_2,\ldots,x_p)$, satisfying $1 \leq x_j < x_{j+1} \leq N$ and $\eta_{x_j}=1$ for all $j$. In the coordinate Bethe ansatz, the eigenvector $\phi$ is presented as a linear combination of plane waves
\begin{equation}
	\label{eq: BA}
	\phi(x_1, \dots, x_p) :=
	\sum_{\sigma \in S_p} A_{\sigma} z_{\sigma(1)}^{x_1} \dots z_{\sigma(p)}^{x_p},
\end{equation}
where $z_1, \dots, z_p$ are non-zero complex numbers, and the sum is over all permutations $\sigma$ in the symmetric group $S_p$ of degree $p$. Substituting \eqref{eq: BA} into the eigenvector equation $M_{\gamma} \phi = \lambda(\gamma) \phi$ yields
	\begin{equation} \label{eq: eigvalue}
    \lambda(\gamma)=e^\gamma \sum_{j=1}^p \frac{1}{z_j} - p,
	\end{equation}
and $z_1$, \dots, $z_p$ satisfy the following system of Bethe ansatz equations (BAE) derived in \cite{derrida1998exactLDfunction}
    \begin{equation} \label{eq: BAE}
		z_k^{N}	= (-1)^{p-1} \prod_{j=1}^{p} 
		\frac{e^\gamma - z_k}{e^\gamma - z_{j}}, \qquad k = 1, \dots, p.
	\end{equation}
    Define the left translation operator $T$ acting by shifting all particle coordinates by one site to the left. This operator permutes the components of an eigenvector according to
    $$
    T\phi(x_1,\ldots,x_p) := \phi(x_1-1,\ldots,x_p-1).
    $$
    From \eqref{eq: BA}, $\phi$ is also an eigenvector of $T$ with eigenvalue $\prod_{j=1}^p z_j$, i.e., 
    $$T \phi = (\prod_{j=1}^p z_j) \phi.$$ 
    Periodic boundary condition implies $T^N\phi = \phi$, which leads to the quantization condition.
    \begin{equation} \label{def momentum}
		\prod_{j=1}^p z_j = e^{\frac{2 \pi i}{N} m} \quad \text{for some }m \in\{0,\dots N-1\}.
	\end{equation}
   The total momentum is quantized in units of $2 \pi/N$.
    
    \begin{remark}
    Equation \eqref{def momentum} can alternatively be derived by multiplying all $p$ equations in \eqref{eq: BAE}, whose right-hand sides reduce each other. Therefore, if we take the $N$th power root of both parts, we obtain a root of unity on the right-hand side.
    \end{remark}

    \subsection{Bethe equations decoupling property} \label{decoupling}
Following \cite{GolinelliMallick2005spectral}, we introduce a new system of coordinates 
$$
Z_j := 2 e^{\gamma}z_j^{-1} - 1, \quad j=1,\ldots,p.
$$
In terms of the $Z_j$, the Bethe ansatz equations \eqref{eq: BAE} become
\begin{equation}  \label{eq: BAE in Zk}
	(1-Z_k)^p(1+Z_k)^{N-p} = -2^N e^{\gamma N} \prod_{j=1}^p \frac{Z_j -1 }{Z_j+1},  \qquad k= 1, \dots, p.
\end{equation}
Denote by $C_{N,\gamma}:=C_{N,\gamma}(Z_1,\ldots,Z_p)$ the right-hand side of each equation, which is independent of index $k$.
\begin{equation} \label{eq: def C-N-gamma}
    C_{N,\gamma}(Z_1,\ldots,Z_p) := -2^N e^{\gamma N} \prod_{j=1}^p \frac{Z_j -1 }{Z_j+1}.
\end{equation}
The periodic boundary condition \eqref{def momentum} translates to
    \begin{equation} \label{eq: momentum in Zk}
        \prod_{j=1}^p (1+Z_j) = 2^pe^{\gamma p}e^{-\frac{2\pi i}{N} m} \quad \text{for some }m \in\{0,\dots N-1\}.
    \end{equation}
The eigenvalue takes a particularly simple form
	\begin{equation} \label{eq: ev via Z}
		\lambda(\gamma)= \frac{1}{2}\sum_{j=1}^p (Z_j-1).
	\end{equation}

    To analyze the system  \eqref{eq: BAE in Zk}, we adopt the approach  of \cite{GolinelliMallick2005spectral}.
    For each $C\in\mathbb{C}$, denote by $u_1$, \dots, $u_N$ the $N$ complex roots of the polynomial equation
	  \begin{equation} \label{eq: pol_eq}
	  	(1-u)^p(1+u)^{N-p} = C.
	  \end{equation}
    The rigorous definition of $u_j$ as a function of $C$ is stated later in Section \ref{geo-bethe}.
    To construct a particular solution of \eqref{eq: BAE in Zk}, we select a subset
    \begin{equation} \label{choice}
    A \subseteq \{1, \dots, N\}, \quad \text{such that} \quad |A| = p,
    \end{equation}
    and impose the consistency condition arising from \eqref{eq: BAE in Zk}
    \begin{equation}  \label{eq: consistency}
	  2^N e^{\gamma N} \prod_{j \in A} \frac{u_j(C_{N,\gamma}) - 1}{u_j(C_{N,\gamma}) + 1} + C_{N, \gamma} =0.
	\end{equation}
    Let $l: \{1,\ldots,p\} \rightarrow A$ be the strictly increasing bijection enumerating the elements of $A$. Then, the corresponding solution is 
    \begin{equation} \label{eq: sol-BAE}
    Z_j := u_{l(j)}(C_{N,\gamma}), \qquad j=1,2,\ldots,p.
    \end{equation}
    The corresponding eigenvalue is \eqref{eq: ev via Z}.

    \begin{remark} \label{rem: iwaomotegi}
     Recently, S.~Iwao and K.~Motegi (see Corollary 2.8 in \cite{2025IwaoMotegiCompleteness}) have rigorously proved that the Bethe ansatz equations possess exactly $\binom{N}{p}$ solutions counted with multiplicities. Therefore, there are enough solutions to potentially generate a full set of eigenvalues for the tilted Markov generator. However, this does not imply that the Bethe ansatz produces a basis of eigenvectors for the entire space. Moreover, it does not imply that every choice of $p$ Bethe roots yields a unique eigenvector. 
    \end{remark}
%
%
%

      \subsection{Geometry of Bethe roots} \label{geo-bethe}
      In this subsection, we precisely define the solution maps for equation \eqref{eq: pol_eq} and analyze the geometric distribution of its roots in the complex plane.
      Fix positive integers $N$ and $p$, and define the density $\rho:=p/N$.  We parametrize the right-hand side of \eqref{eq: pol_eq} as $C = -(r e^{i \theta})^p$, where
      \begin{equation}
         r := |C|^{\frac 1p}, \qquad \theta := \frac{1}{p} \arg (-C),
      \end{equation}
      with $\arg(\cdot)$ denoting the \emph{principal value} of a complex number in $(-\pi,\pi]$.
      Hence, $\theta \in (-\frac{\pi}{p}, \frac{\pi}{p}]$.

      Since \eqref{eq: pol_eq} is a polynomial of degree $N$, it has $N$ distinct roots for all $r \notin \{0, r_\mathrm{cr}\}$, where
      \begin{equation} \label{eq: r-cr}
      r_\mathrm{cr} := 2^{\frac1\rho}\rho(1-\rho)^{\frac1\rho-1}.
      \end{equation}
      
      To label the roots uniquely, we define a function $F:\mathbb{C} \rightarrow (-N,N]$ that associates to each root a unique value based on its argument
      $$
      F(u) := \frac{p}{\pi}\arg(u - 1) + \frac{N-p}{\pi}\arg(u + 1).
      $$
      If $u$ is a root of \eqref{eq: pol_eq}, then
      $$
      F(u) = \frac{\arg C}{\pi} + p + 2j \quad \text{for some} \quad j \in \mathbb{Z},
      $$
      allowing us to index roots by an interval containing $F(u)$. 
      
      To be precise, let $\tau: \mathbb{R} \to \mathbb{R}/(2N\mathbb{Z}) \cong (-N, N]$ be the quotient map and let intervals be 
      $$I_j := \tau\big( (-p + 2(j-1), -p + 2j] \big) \subset (-N, N], \quad j=1,\dots, N. $$
      If $\theta \not= 0$, i.e. $C \notin \mathbb{R}_{\ge0}$, we define $u_j=u_j(C)$ to be the unique solution to \eqref{eq: pol_eq} such that $F(u_j) \in I_j$. If $C\in\mathbb{R}_{>0}$, we define $u_j(C)$ as a continuous extension from $\Im(C) < 0$
      $$
      u_j(C) := \lim_{\delta \uparrow 0} u_j(Ce^{i \delta}).
      $$

\begin{example} \label{half-fill}
In the case of half-filling $N=2p$, the root maps $u_j=u_j(C)$ defined above are explicitly given by
$$
\begin{aligned}
    &u_j(C) = \sqrt{1-(-C)^{\frac{1}{p}}e^{i\frac{(2j-1)\pi}{p}}}, \qquad \text{for } j=1, \ldots, p,
\end{aligned}
$$
and $u_j=-u_{j-p}$ for $j=p+1$, \dots, $2p$.
\end{example}

We further illustrate the geometrical distribution of the roots of the equation  \eqref{eq: pol_eq}.
    All roots $u_j$ are located on  \emph{Cassini ovals} defined by
    \begin{equation} \label{eq: Cassini contours}
    \mathcal{C}(r) := \big\{u \in \mathbb{C}; \ |1-u|^{\rho}|1+u|^{1-\rho}= r^{\rho}\big\}.
	\end{equation}
    This curve has foci at $\pm 1$ and undergoes topological changes at a focal distance $r = r_\mathrm{cr}$ defined in \eqref{eq: r-cr}. Below, we summarize the topology of the curve and the distribution of the roots.
    \begin{itemize}
        \item[\romannumeral1.] If $r=0$, the curve $\mathcal{C}(r)$ degenerates, containing a root $u=1$ with multiplicity $p$ and another root $u=-1$ with multiplicity $N-p$.
        \item[\romannumeral2.] If $r\in(0,r_\mathrm{cr})$, there are two disjoint ovals in $\mathcal{C}(r)$, centred at $1$ and $-1$, respectively. The right oval contains $p$ distinct roots $u_1$, \dots, $u_p$ while the left contains the remaining $N-p$ distinct roots.
        \item[\romannumeral3.] If $r=r_\mathrm{cr}$, $\mathcal{C}(r)$ forms a deformed lemniscate of Bernoulli with a double point at $1-2\rho$. In particular, if $C=r_\mathrm{cr}^p$, i.e. $\theta=0$, the point $u_p=u_N=1-2\rho$ is a double root. If $\theta \neq 0$, the roots $u_1, \dots, u_p$ lie on the right lobe of the lemniscate, and $u_{p+1}, \dots, u_N$ on the other.
        \item[\romannumeral4.] If $r\in(r_\mathrm{cr},+\infty)$, $\mathcal{C}(r)$ forms a single oval that encompasses $\pm1$. All $N$ roots of the polynomial equation \eqref{eq: pol_eq} are distinct and lie on this curve.
    \end{itemize}
    
    For an illustration of the roots of the equation \eqref{eq: pol_eq} on Cassini ovals, see Figure \ref{fig: Asymmetry of Cassini contours}. This figure compares the case of half-density Cassini ovals that are symmetric to the origin with the case of arbitrary density $\rho$ with asymmetric Cassini ovals. For an example of labelling the roots, refer to Figure \ref{fig:  Cassini contour with sep curve}, the right column.

    Furthermore, the first $p$ roots $u_1$, \dots, $u_p$ are always contained in the following region:
    \begin{equation} \label{eq: domain}
        D_+ := \big\{ u\in\mathbb{C};\, \rho |\arg(1-u)| \ge (1-\rho) |\arg(1+u)| \big\}.
    \end{equation}
    \begin{example}
        In the case of half-filling $N=2p$, we have $r_\mathrm{cr}=1$ and $D_+ = \{u \in\mathbb{C} \mid \ \Re u\ge0\}$.
    \end{example}

    \begin{remark}
        If $r \ge r_\mathrm{cr}$ and $\theta=0$, the roots $u_p$ and $u_N$ would appear at the boundary $\partial D_+$. Otherwise, $\{u_1, \ldots, u_p\} \subset D_+^\circ$.
    \end{remark}

      \begin{figure}
        \center{\includegraphics[width=1\linewidth]{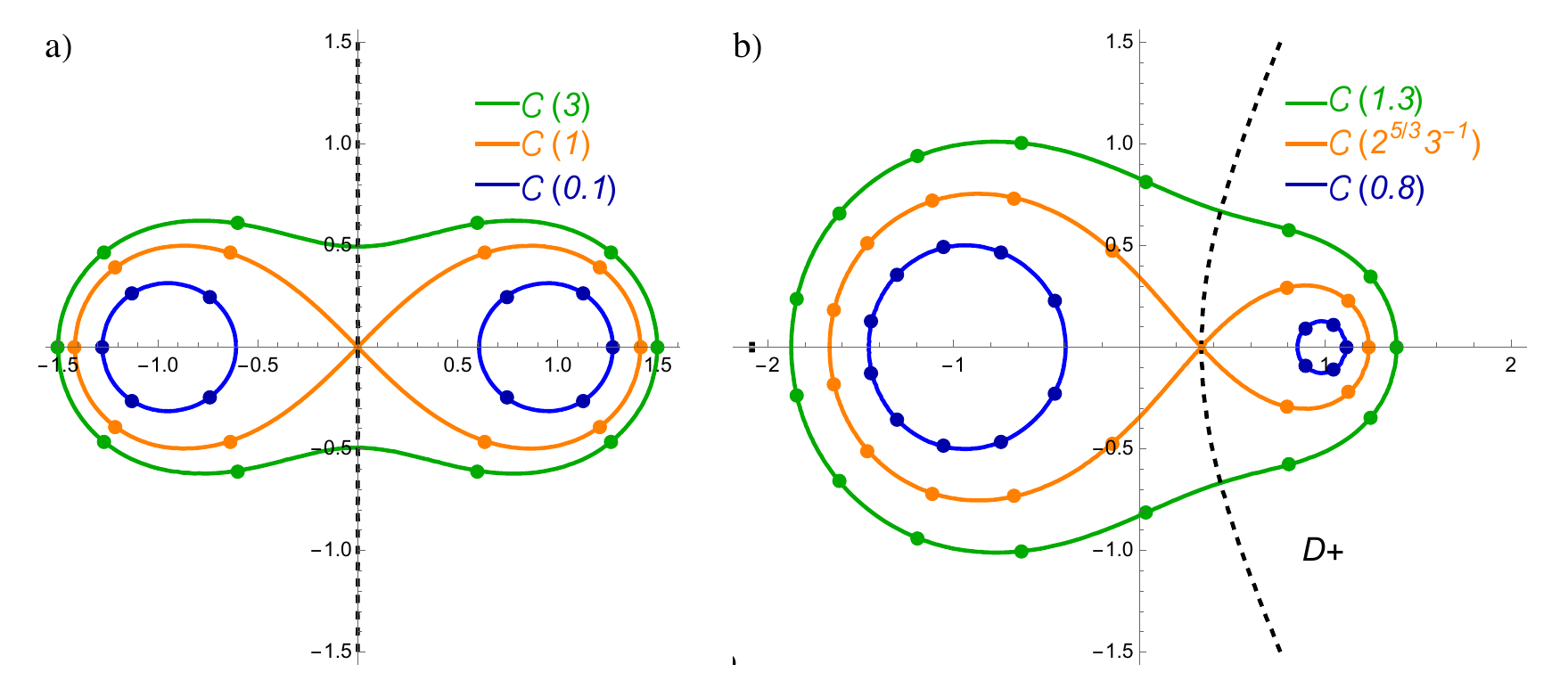}}
		\caption{a) Cassini ovals for (a) $N=10$, $\rho = 1/2$ (left) and for (b) $N=15,$ $p=5$, $\rho = 1/3$ (right). The green single loop corresponds to $r > r_{\mathrm{cr}},$ the orange deformed lemniscate of Bernoulli corresponds to $r = r_{\mathrm{cr}}$, and the two blue ovals correspond to $r <r_{\mathrm{cr}}$. The $N$ solutions of \eqref{eq: pol_eq} are marked by the dots. The dashed line is the boundary of the domain $D_+$.
		}
		\label{fig: Asymmetry of Cassini contours}
    \end{figure}

\subsection{Assumptions on Bethe root selection} \label{subsec: ass}
    
    Recall the strategy introduced in Section \ref{decoupling}: to find a particular solution to equation \eqref{eq: BAE in Zk}, one selects the $p$ solution maps from a total of $N$. This selection is specified by a subset of indexes $A$ \eqref{choice}, or, equivalently, by a choice function $l$. Once $l$ is fixed, the consistency equation \eqref{eq: consistency} determines $C_{N,\gamma}$, and the solution set $Z_1, \dots, Z_p$ is given by \eqref{eq: sol-BAE}. In this subsection, we focus on assumptions regarding the selection of Bethe roots corresponding to the largest and second-largest eigenvalues. 

 \begin{proposition} \label{prop: Creal}
        Let $\{Z_1,\ldots,Z_p\}$ be a solution to \eqref{eq: BAE in Zk} that results in the largest eigenvalue $\lambda_1(\gamma)$ of $M_\gamma$. Then, the corresponding $C_{N,\gamma}$ is real.
	\end{proposition}
    
    \begin{proof}
        If a solution $\{Z_1,\ldots, Z_p\}$ is invariant under complex conjugation, we find $C_{N,\gamma}$ from \eqref{eq: def C-N-gamma} to be real. Assume the set $\{Z_1,\ldots, Z_p\}$ is not invariant under complex conjugation. Then, its complex conjugate set $\{\overline{Z}_1,\ldots,\overline{Z}_p\}$ is another solution to \eqref{eq: BAE in Zk} yielding a complex conjugate eigenvalue $\overline{\lambda_1(\gamma)}$. Since the largest eigenvalue is real and simple by the Perron-Frobenius theorem, this contradicts the assumption.
    \end{proof}

For finite $N$, the largest eigenvalue $\lambda_1(\gamma)$ of $M_{\gamma}$ is simple according to the Perron-Frobenius theorem, and therefore is an analytic function of $\gamma$. At $\gamma=0$, $\lambda_1(0)=0$ corresponds to the degenerate case of the focal radius $|C_{N,\gamma}|=0$, where the rightmost $p$ roots are $Z_j=u_j(0)=1$ for $j=1, \dots, p$.
Consequently, there exists some positive $\epsilon_N$, such that for $|\gamma|<\epsilon_N$, the largest eigenvalue $\lambda_1(\gamma)$ is determined by the  focal radius $|C_{N,\gamma}|<r_\mathrm{cr}^p$ and the root selection given by
\begin{equation} \label{choice1}
    A=\{1,\ldots,p\}.
\end{equation}
This choice corresponds to selecting the $p$ roots with the largest real part on the right oval of the Cassini curve, yielding the largest eigenvalue.

In the seminal paper \cite{GwaSpohn1992spgap} by L.H.~Gwa and H.~Spohn, the authors conjectured that $\lambda_2(\gamma)$ is given by the minimal modification of \eqref{choice1}, i.e. keeping the $p-1$ roots in \eqref{choice1} with the largest real part, and replacing the remaining one by its nearest neighbour:
\begin{equation} \label{choice2}
    A = \{1,\ldots,p-1,p+1\} \quad \text{or} \quad \{2,\ldots,p,N\}.
\end{equation}
\begin{remark} 
    Since $\theta$ belongs to an interval $\theta \in (-\frac{\pi}{p}, \frac{\pi}{p}]$ shrinking as $p$ grows, at least asymptotically, alternative choices  
    \begin{align}
A &= \{ 1, \ldots, p-1, N\} \label{eq:first-choice} \\
&\quad \text{or} \quad \{2,\ldots, p, p+1\} \label{eq:second-choice}
\end{align}
    produce eigenvalues with the same real part as those in \eqref{choice2}, potentially corresponding to complex conjugate pairs.
\end{remark} 

Although these root selections hold near $\gamma$ sufficiently close to $0$, for other values of $\gamma$ eigenvalue crossing or merges can occur, causing the solution to \eqref{eq: BAE in Zk} to switch to a different choice function $l$. The following proposition provides a simple counterexample for negative $\gamma$.

\begin{proposition} \label{prop: CCnoRoots}
    If $\gamma < \log(1-\rho)$, the consistency equation \eqref{eq: consistency} has no solution under the choice \eqref{choice1}.
\end{proposition}

\begin{proof}
Expressing $|u-1|$ from a Cassini curve equation \eqref{eq: Cassini contours}, we simplify the consistency condition \eqref{eq: consistency} to show
\begin{equation}
   (2-2\rho)^p < \min_j |1+Z_j|^p < \prod_{j=1}^p|1+Z_j| = 2^p e^{\gamma p}.
\end{equation}
The leftmost inequality follows from the location of the leftmost point in the domain $D_+$ delivering the minimal possible value for $|1+Z_j|.$
Therefore, if $\gamma < \log(1-\rho)$, there is no solution.
\end{proof}

For the reason stated above, we formulate the following assumptions.

\begin{assumption}[Choice of roots for $\gamma>0$] \label{ass: gammapos}
    Let $\gamma>0$. For sufficiently large $N$, the largest eigenvalue $\lambda_1(\gamma)$ is achieved by the choice in \eqref{choice1}, while the second-largest one $\lambda_2(\gamma)$ is given by the minimal modification \eqref{choice2}.
\end{assumption}

\begin{assumption}[Choice of roots for $\gamma<0$] \label{ass: gammaneg}
    Let $\gamma<0$. For sufficiently large $N$, the largest eigenvalue $\lambda_1(\gamma)$ is achieved by the choice in \eqref{eq:first-choice}, while the second-largest one $\lambda_2(\gamma)$ is given by further modifications
    \begin{equation} \label{choice3}
        \begin{aligned}
            A = \{1,\ldots,p-1,p+2\} \quad &\text{or} \quad \{1,\ldots,p-2,p,p+1\}\\
            \text{or} \quad \{2,\ldots,p,N-1\} \quad &\text{or} \quad \{1,3\ldots,p,N\}.
        \end{aligned}
    \end{equation}
\end{assumption}
\begin{remark}
    Although the choice of the roots in the assumptions above involves different selection rules for positive and negative $\gamma$, the selection procedures are continuous extensions of each other. This is consistent with the results of B. Derrida and C. Appert \cite{1999Derrida1Appert}, where the eigenvalue structure and corresponding root configurations evolve continuously as the parameter $\gamma$ varies.
\end{remark}
The following \emph{a priori} estimate of $C_{N,\gamma}$ corresponding to $\lambda_1(\gamma)$ is a direct consequence of the assumptions above.

\begin{proposition} \label{prop: C}
    If $\gamma>0$, the Bethe roots $\{Z_1,\ldots,Z_p\}$ corresponding to $\lambda_1(\gamma)$ satisfy \eqref{choice1}, $C_{N,\gamma} \in (-\infty, r_\mathrm{cr}^p]$.
    If $\gamma<0$ and a set $\{Z_1,\ldots,Z_p\}$ corresponding to $\lambda_1(\gamma)$ satisfies \eqref{choice2}, $C_{N,\gamma} \in (0,r_\mathrm{cr}^p)$.
\end{proposition}

\begin{proof}
\textit{Case: $\gamma > 0.$} 
   We proceed by contradiction.
Assume $C_{N,\gamma} > r_{\text{cr}}^p$. For the polynomial equation \eqref{eq: pol_eq}:
    \begin{itemize}
        \item $p$ odd: No real positive root exists; all roots are complex conjugate pairs.
        \item $p$ even: One real root $>1$ exists; the others are complex conjugate pairs.
    \end{itemize}
    In both cases, the $p$ roots with the largest real parts are not uniquely determined, contradicting the simplicity of the largest eigenvalue. Thus, $C_{N,\gamma} \leq r_{\text{cr}}^p$.

\textit{Case: $\gamma < 0.$} Assume $C_{N,\gamma} \notin (0, r_{\text{cr}}^p)$. If $C_{N,\gamma} < 0$
        \begin{itemize}
            \item $p$ odd: A real root $>1$ exists; others are complex conjugate pairs.
            \item $p$ even: No real positive root; all are complex conjugate pairs.
        \end{itemize}
        
In both, the selection of $p$ roots is not unique, violating the simplicity of the largest eigenvalue.
If $C_{N,\gamma} > r_{\text{cr}}^p$, the Cassini oval merges into a single component, again making the selection non-unique.
Therefore, $C_{N,\gamma} \in (0, r_{\text{cr}}^p)$.

\end{proof}

\begin{figure}
    \center{\includegraphics[width=1\linewidth]{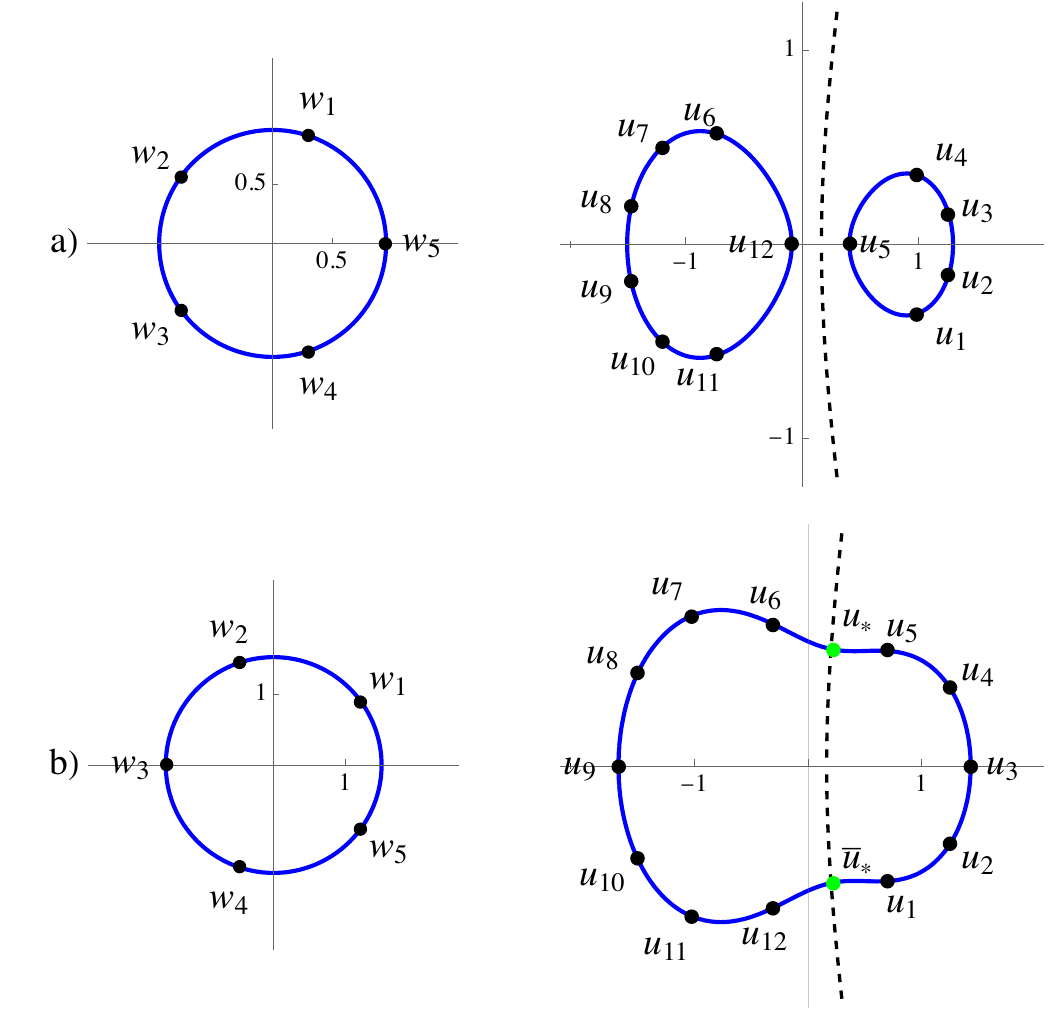}}
			\caption{Illustration for $N=12, p=5 \ \rho = 5/12$ showing the chosen labelling of points $w_1, \dots, w_p$ (left column) on a circle and the corresponding points $u_1, \dots, u_N$ on Cassini curves. 
            a) Subcritical focal radius is chosen with $\theta = 0$. The Cassini curve consists of two ovals.
            b) Supercritical focal radius with $\theta = \pi/p$.
             The Cassini curve consists of one oval.
             A dashed line is the boundary of the domain $D_+$ defined by \eqref{eq: domain}.}
			\label{fig:  Cassini contour with sep curve}
	\end{figure}

\section{Asymptotic Analysis} \label{sec: Asymptotic analysis}
In this section, we analyze the asymptotic behaviour of the largest and second largest eigenvalues $\lambda_1(\gamma)$ and $\lambda_2(\gamma)$ in the thermodynamic limit, under Assumption \ref{ass: gammapos} and \ref{ass: gammaneg} in the previous section.

To approximate sums over $Z_j = u_{l(j)}(C)$ with $j=1, \dots, p$, we use the Euler-Maclaurin formula, which allows us to replace discrete sums with integrals plus controlled error terms. We use the version of the Euler-Maclaurin formula (\cite{Dieudonne1971}, Chapter IX, problem 10, p. 285),  which expresses the error term as an integral, making it more suitable for our analysis than the traditional sum over derivatives. Throughout, all logarithms and arguments are taken in the principal branch, i.e., $\arg z \in (-\pi, \pi]$ and $\log z = \log|z| + i\arg z$.

\begin{theorem} \label{Euler-Maclaurin sum theorem}
    Let $f(z)$ be a function analytic in a strip $\alpha< \Re z < \beta $ and let $m,n$ be two integers such that $\alpha<m<n<\beta$. Suppose that 
    $$ \lim_{t \rightarrow \pm \infty} e^{-2 \pi |t|} f(s+ i t)=0$$
    uniformly for $s \in (\alpha, \beta)$.
    Then
	\begin{equation}\label{eq: EM summation}
		\sum_{j = m}^n f(j) = \int_m^n f(z) \ \ dz + E_f(m,n),
	\end{equation}
	where the error term is given by
    \begin{equation} \label{eq: EM error term}
        \begin{aligned}
        &E_f(m,n) := \frac{f(m)+f(n)}2\\
        &\qquad \qquad \quad + \int_0^\infty \frac{f(n+it)-f(n-it)-f(m+it)+f(m-it)}{i(e^{2\pi t}-1)} \, dt.
        \end{aligned}
    \end{equation}
\end{theorem}

Before applying the Euler-Maclaurin formula, we note that the solution maps $u_j = u_j(C)$ for $j=1, \dots, p$ can be parametrized by points lying on a circle of radius $|C|^{\frac 1p}$. This lemma generalizes the result \cite[(4.7)--(4.8)]{GwaSpohn1992spgap}, which was stated for the special case $\rho=1/2$,  to arbitrary density $\rho$.

\begin{lemma} \label{lem: implicit}
    There exists an invertible analytic function
    $$
    v: \mathbb{C} \setminus [r_\mathrm{cr},+\infty) \rightarrow D_+^\circ,
    $$
    reconstructing the positions of $u_j(C) = v(w_j)$ in a complex plane from equidistant points of a circle.
    $$
    w_j := (-C)^{\frac{1}{p}} e^{i\frac{(2j-1)\pi}{p}}, \qquad j=1,\ldots,p.
    $$
\end{lemma}

\begin{proof}
    Let $w$ be the complex-valued function
    \begin{equation} \label{def: w(u)}
        w(u) := (1-u) \ (1+u)^{\frac{N}{p}-1}.
    \end{equation}
    It is easy to check that $w(u_j)=w_j$ for the first $p$ solutions $j=1$, \dots, $p$. Since $w$ is one-to-one on the domain $D_+^\circ$ and the image is $\mathbb{C}\setminus[r_\mathrm{cr},+\infty)$, we define $v=v(w)$ to be the inverse of $w|_{D_+^\circ}$.
\end{proof}
%
\begin{remark}
    Recall $r = |C|^{\frac 1p}$. If $r>r_\mathrm{cr}$,
    $$
    \begin{aligned}
        &\lim_{N\rightarrow\infty} u_p(C) = \lim_{\delta\downarrow0} v(re^{-i\delta}) = u_*,\\
        &\lim_{N\rightarrow\infty} u_1(C) = \lim_{\delta\downarrow0} v(re^{i\delta}) = \bar u_*,
    \end{aligned}
    $$
    where $u_*=u_*(r)$ is the unique solution to
    \begin{equation} \label{eq: def-u*}
        \left\{
        \begin{aligned}
            &\,(1-u_*)^\rho(1+u_*)^{1-\rho} = r^\rho,\\
            &\,\rho \arg(1-u_*) + (1-\rho) \arg(1+u_*) = 0,
        \end{aligned}
        \right.
    \end{equation}
    such that $\Im(u_*)>0$, and $\bar u_*$ is the complex conjugate of $u_*$. Notice that the Cassini contour $\mathcal C(r)$ intersects the boundary of $D_+$ at $(u_*,\bar u_*)$, see in Figure \ref{fig:  Cassini contour with sep curve}.
    For $r=r_\mathrm{cr}$, we adopt the convention $u_*(r_\mathrm{cr}) = \bar u_*(r_\mathrm{cr}) = 1-2\rho$.
\end{remark}
\begin{example} \label{ex:half-fill_r*}
    In the case of half-filling $N=2p$, we find a purely imaginary
    \begin{equation} \label{half-fillr*}
        u_*(r) = i \sqrt{r-1}.
    \end{equation}
\end{example}

\subsection{Asymptotic analysis for $\gamma>0$}
\label{sec: asym+}

Introduce the thermodynamic limit functions
$$
\begin{aligned}
    G_\infty(z) &:= \frac{1}{\pi} \bigg\{ \frac{1-\rho}{\rho} \log|1+z| \arg(1+z)\\
    &\qquad + \log2 \arg(1-z) - \Im \left[ \operatorname{Li}_2 \left( \frac{1}{2} - \frac{z}{2} \right) \right] \bigg\},\\
    \Lambda_\infty(z) &:= \frac{1}{\pi} \left\{ \frac{\Im z}{2} + (\rho-1)\arg(1+z) \right\},\\
    g_\infty(z) &:= \frac{\pi i}{z-1+2\rho} \left\{ 1-z^2 - \frac{2\rho\Im z}{\arg(1+z)} \right\},
\end{aligned}
$$
where in $G_\infty$, $\operatorname{Li}_2$ is the \emph{dilogarithm function}
\begin{equation}
    \operatorname{Li}_2(z) := - \int_0^z \frac{\log(1-u)}{u}\ du, \qquad z \in \mathbb{C} \setminus [1, +\infty).
\end{equation}
Theorem \ref{thm: asym+}, the main result of this section, states that for $\gamma>0$, these functions capture the thermodynamic behaviour of the consistency condition, the largest eigenvalue, and the spectral gap, respectively.

\begin{theorem} \label{thm: asym+}
    For $\gamma>0$, the consistency condition of the largest and the second largest eigenvalue takes the following form in the thermodynamic limit.
    \begin{equation} \label{eq: asym-cc}
    \gamma = G_\infty\big(u_*(r_*)\big).
    \end{equation}
    This equation has a unique solution $r_* = r_*(\gamma) > r_\mathrm{cr}$. Under Assumption \ref{ass: gammapos},
    \begin{align}
    &\limsup_{N\rightarrow\infty} \Big|\lambda_1(\gamma;N) - N\Lambda(\gamma, \rho) \Big| < \infty, \label{eq: asym-1+}\\
        &\lim_{N\rightarrow\infty} \Big|N\big(\lambda_1(\gamma;N) - \lambda_2(\gamma;N)\big) - g(\gamma, \rho)\Big| = 0, \label{eq: asym-gap+}
    \end{align}
    where 
    \begin{equation} \label{def LambdaCoef}
        \Lambda(\gamma, \rho): = \Lambda_\infty\big(u_*(r_*(\gamma))\big),
    \end{equation}
    \begin{equation}\label{def gCoef}
        g(\gamma, \rho): = g_\infty\big(u_*(r_*(\gamma))\big).
    \end{equation}
    In addition, $r_*(\gamma)$ is strictly increasing in $\gamma$ and $\lim_{\gamma\rightarrow0+} r_*(\gamma) = r_\mathrm{cr}$.
\end{theorem}

In the following two subsections, we prove this theorem, presenting Proposition~\ref{prop: asym} and the Proofs of  \eqref{eq: asym-1+} and \eqref{eq: asym-gap+}.
Figure \ref{fig: GLambda} illustrates the monotonic behaviour of $G_\infty\big(u_*(r)\big)$ and $\Lambda_\infty\big(u_*(r)\big)$ as functions of $r$.

\begin{figure}
    \center{\includegraphics[width=0.7\linewidth]{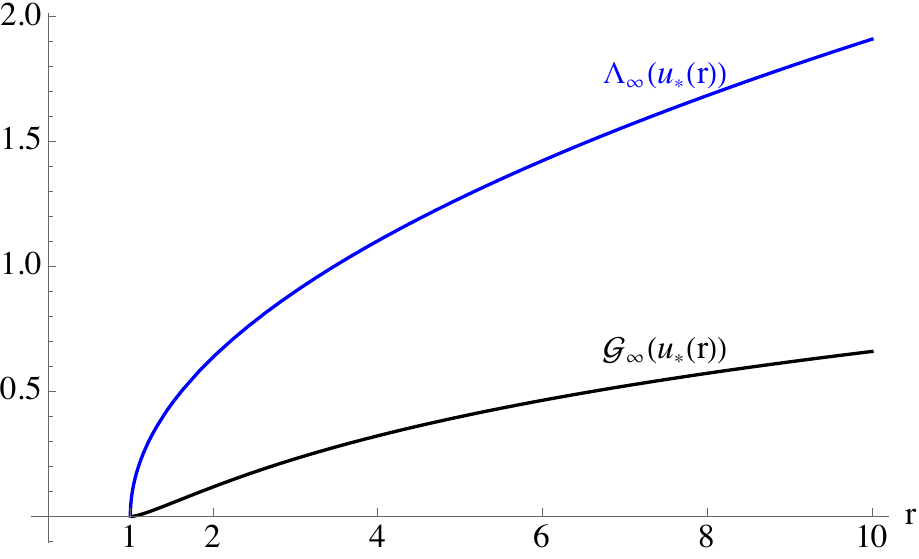}}
			\caption{The graphs of $G_\infty\big(u_*(r)\big)$ and $\Lambda_\infty\big(u_*(r)\big)$ for $\rho = \frac 12 $ illustrated in a variable $r.$ Here, $r_\mathrm{cr} = 1$.}
			\label{fig: GLambda}
	\end{figure}

%
%

\subsubsection{The largest eigenvalue}
From the strategy described above to find the largest eigenvalue we select $p$ solution maps out of $N$, i.e. fix a set $A$. According to Assumption \ref{ass: gammapos}, $A = \{1,\ldots,p\}$ for $\lambda_1(\gamma)$. Then, we solve the consistency condition \eqref{eq: consistency}, finding $C_{N,\gamma}$, and use \eqref{eq: ev via Z}--\eqref{eq: sol-BAE} to calculate the eigenvalue.

As Proposition~\ref{prop: C} states, the value $C_{N,\gamma}$ has to be real, such that $C_{N,\gamma} \in (-\infty, r_\mathrm{cr}^p]$. Therefore, we introduce 
\begin{equation} \label{eq: def-r-1}
    r_N = r_N(\gamma) := |C_{N,\gamma}|^{\frac{1}{p}} \in \mathbb{R}_{\ge0}.
\end{equation}
We show that $\{r_N;N\ge1\}$ is a convergent sequence, and we characterize the limit $r_*=r_*(\gamma)$ by the asymptotic equation \eqref{eq: asym-cc} arising from \eqref{eq: consistency}.

The choice of the solutions maps is
$$
Z_j = u_j(C_{N,\gamma}), \qquad j=1,\ldots,p.
$$
In proving Proposition~\ref{prop: Creal}, we argued that the set $\{Z_1,\ldots, Z_p\}$ for the largest eigenvalue must be invariant under complex conjugation. Therefore, the momentum $m=0$ in \eqref{eq: momentum in Zk}. From the definition of the mapping $w(\cdot)$ with $w(Z_j) = w_j$ used in Lemma \ref{lem: implicit}, we find
$$
\prod_{j=1}^p \frac{1-Z_j}{1+Z_j} \prod_{j=1}^p (1+Z_j)^{\frac{N}{p}} = \prod_{j=1}^p w_j = (-1)^{p-1}C_{N,\gamma}.
$$
This observation simplifies the consistency equation \eqref{eq: consistency}, which becomes
\begin{equation} \label{eq: cc1 simple}
	2^p e^{\gamma p} = \prod_{j=1}^{p} \big(1+u_j(C_{N,\gamma})\big).
\end{equation}
\begin{remark}
    This form of the consistency condition agrees with the quantization condition \eqref{def momentum}.
\end{remark}
Hence, $C_{N,\gamma}$ is determined by the equation
\begin{equation} \label{eq: log cc1}
	\gamma = \frac{1}{p} \sum_{j=1}^{p} \log \big(1+u_j(C_{N,\gamma})\big) - \log2.
\end{equation}
Once $C_{N,\gamma}$ is found, the largest eigenvalue follows from \eqref{eq: ev via Z}
\begin{equation} \label{eq: la1}
	\lambda_1(\gamma) = \frac12\sum_{j=1}^{p} \big(u_j(C_{N,\gamma})-1\big).
\end{equation}
Introducing a variable $C \in (-\infty,r_\mathrm{cr}^p]$, we denote the right-hand sides of \eqref{eq: log cc1} and \eqref{eq: la1} by
\begin{align} \label{def GN}
    G_N(C) &:= \frac{1}{p} \sum_{j=1}^{p} \log \big(1+u_j(C)\big) - \log2,\\
    \Lambda_N(C) &:= \frac12\sum_{j=1}^{p} \label{def LN}\big(u_j(C)-1\big).
\end{align}
The next proposition is crucial for proving the first assertion \eqref{eq: asym-1+} in Theorem~\ref{thm: asym+}.

\begin{proposition} \label{prop: asym}
    Let $r=|C|^{\frac 1p}$. Then,
    \begin{align}
        &G_N(C) = \mathbf1_{\{r>r_\mathrm{cr}\}} G_\infty\big(u_*(r)\big) + \frac{R_{N,C}}N, \label{eq: asym-g}\\
        &\Lambda_N(C) = N\mathbf1_{\{r>r_\mathrm{cr}\}} \Lambda_\infty\big(u_*(r)\big) + R'_{N,C}, \label{eq: asym-la}
    \end{align}
    where $R_{N,C}$ and $R'_{N,C}$ are uniformly bounded for any compact set of $C$.
\end{proposition}

\begin{proof}
    Recall the function $v$ defined in lemma \ref{lem: implicit} and let
    $$
    f(z) := \log (1+v((-C)^{\frac{1}{p}}e^{i\frac{(2z-1)\pi}{p}})).
    $$
    Then $f(j) = \log(1+u_j(C))$ for $j=1$, \dots, $p$.
    We divide the region of $C$ into three parts according to the topology of Cassini ovals: (\romannumeral1) supercritical  $C \le -r_\mathrm{cr}^p$, (\romannumeral2) subcritical  $|C| < r_\mathrm{cr}^p$, and  (\romannumeral3) critical $C = r_\mathrm{cr}^p$.
    
    In case (\romannumeral1), by lemma \ref{lem: implicit}, $f$ is differentiable in $x \in [1,p]$. Hence, Theorem~\ref{Euler-Maclaurin sum theorem} yields that
    $$
    G_N(C) = \frac{1}{p} \int_1^p f(z) \ dz + \frac{E_f(1,p)}{p} - \log2.
    $$
    Recall the Cassini contour $\mathcal{C}(r)$ defined by \eqref{eq: Cassini contours}. Applying Lemma \ref{lemB1: change of variables},
    $$
    \frac{1}{p} \int_1^p f(z)\ dz = \frac{1}{2\pi i \rho} \int_{\Gamma(r)} \frac{(u-1+2\rho)\log (1+u)}{u^2-1} \ du,
    $$
    where $\Gamma(r)$ is the path going from $u_1(C)$ to $u_p(C)$ counter-clockwise along $\mathcal{C}(r)$. Observe that $u_1(C)$ and $u_p(C)$ are complex conjugates. Computing the integral with Cauchy residue theorem (see Lemma \ref{lemB2:contour integration}),
    $$
    \frac{1}{p} \int_1^p f(z) \ dz = \log2 + G_\infty\big(u_p(C)\big).
    $$
    Hence,
    $$
    G_N(C) = G_\infty\big(u_p(C)\big) + \frac{E_f(1,p)}{p}.
    $$
    From the definition in \eqref{eq: EM error term}, $E_f(1,p)$ is bounded and continuously dependent on $C$. This, together with Lemma \ref{lemB3: asymptotis of IN} then yields that
    $$
    G_N(C) = G_\infty\big(u_*(C)\big) + O_C\big(N^{-1}\big),
    $$
    where the remainder is continuously dependent on $C$.

    In case (\romannumeral2), $f$ is everywhere differentiable, so
    \begin{equation} \label{eq:GNest_subcr}
         G_N(C) = \frac{1}{p} \int_0^p f(z)\ dz + \frac{E_f(0,p) - f(0)}{p} - \log2.
    \end{equation}
  
    Since $f(0) = f(p)$, the same change of variables as in Lemma \ref{lemB1: change of variables} gives
    $$
    \frac{1}{p} \int_0^p f(z)\ dz = \frac{1}{2\pi i \rho} \oint_{\mathcal{C}_+(r)} \frac{(u-1+2\rho)\log(1+u)}{u^2-1}\ du = \log2,
    $$
    where $\mathcal{C}_+(r)$ is the closed right oval of the Cassini contour \eqref{eq: Cassini contours}, and the last equation follows from the Cauchy residue theorem. The result then follows immediately.
    
    In the remaining case (\romannumeral3), $f$ is \emph{not} differentiable at $u_p(C) = 1-2\rho$, so we can only apply Theorem~\ref{Euler-Maclaurin sum theorem} on $[1,p-1]$ to obtain
    $$
    G_N(C) = \frac{1}{p} \int_1^{p-1} f(z)\ dz + \frac{E_f(1,p-1) + f(p)}{p} - \log2.
    $$
    Since $u_1(C)$ and $u_{p-1}(C)$ are complex conjugates,
    $$
    \frac{1}{p} \int_1^{p-1} f(z)\ dz = \log2 + G_\infty\big(u_{p-1}(C)\big),
    $$
    by Lemma \ref{lemB2:contour integration}.
    As $N\rightarrow\infty$, $u_{p-1}$ converge to $u_*(r_\mathrm{cr})=1-2\rho$, so it follows from \ref{lemB3: asymptotis of IN} that
    $$
    G_N(C) = G_\infty\big(u_*(r_\mathrm{cr})\big) + O_C(N^{-1}) = O_C(N^{-1}).
    $$
    The asymptotic estimate \eqref{eq: asym-g} is then verified. The other assertion \eqref{eq: asym-la} is proved in the same way.
\end{proof}

\begin{proof}[Proof of \eqref{eq: asym-1+} in Theorem~\ref{thm: asym+} ] \label{proof asymG}
By Proposition~\ref{prop:C bounded}, the sequence of $r_N=r_N(\gamma)$ is uniformly bounded, so \eqref{eq: log cc1} and \eqref{eq: asym-g} yield that its limit point $r_*=r_*(\gamma)$ shall satisfy the asymptotic equation
$$
\gamma = \mathbf1_{\{r>r_\mathrm{rc}\}} G_\infty\big(u_*(r_*)\big).
$$
The derivative in $r$ 
$$
\frac{d}{dr} G_\infty\big(u_*(r)\big) = \frac{1}{\pi r} \arg\big(1+u_*(r)\big) + o(1) > 0,
$$
shows that $G_\infty(u_*(r))$ is strictly increasing. Therefore, we conclude that the solution to the asymptotic equation is unique $r_*(\gamma) > r_\mathrm{rc}$.  Therefore, combining \eqref{eq: la1}, \eqref{def LN}, \eqref{eq: asym-la} we obtain the largest eigenvalue expansion \eqref{eq: asym-1+} in the thermodynamic limit.
\end{proof}

\begin{example}
    In the case of half-filling $N=2p$, we substitute $u_*(r(\gamma))$ from \eqref{half-fillr*} to get
    \begin{equation}
    \Lambda(\gamma, \frac 12):= \frac{1}{\pi} \left\{ \frac{\sqrt{r_*(\gamma)-1}}{2} + (\rho-1)\arctan(\sqrt{r_*(\gamma)-1}) \right\}.
\end{equation}
\end{example}

\subsubsection{The spectral gap }

From Assumption \ref{ass: gammapos}, the second largest eigenvalue $\lambda_2(\gamma)$ is given by one of the two possible sets $A$ in \eqref{choice2}. In this section we consider only
$$
A=\{1,2,\ldots,p-1,p+1\},
$$
since the other one can be treated similarly.

With some abuse of notation, let $\tilde{C}_{N,\gamma}$ be the solution to \eqref{eq: consistency} corresponding to $\lambda_2(\gamma)$, and denote
\begin{equation}\label{eq: Zj_tilde}
\tilde{Z}_j = u_j\big(\tilde{C}_{N,\gamma}\big), \qquad j=1,\ldots,p+1,
\end{equation}
as in Lemma \ref{lem: implicit}, the function $w$ then sends these Bethe roots to
$$
w\big(\tilde{Z}_j\big) = \tilde{w}_j := \big(\!-\tilde{C}_{N,\gamma}\big)^{\frac{1}{p}} e^{i\frac{(2j-1)\pi}{p}}, \qquad j=1,\ldots,p.
$$
Also notice that $w(\tilde{Z}_{p+1}) = w_1$, therefore
$$
\prod_{j \in A} \frac{1-\tilde{Z}_j}{1+\tilde{Z}_j} \prod_{j \in A} \big(1+\tilde{Z}_j\big)^{\frac{N}{p}} = \prod_{j \in A} w(\tilde{Z}_j) = (-1)^{p-1}\tilde C_{N,\gamma} e^{\frac{2\pi i}{p}}.
$$
Similarly to \eqref{eq: log cc1}, this simplifies the consistency equation \eqref{eq: consistency} to
\begin{equation} \label{eq: log cc2}
    \gamma = \frac{1}{p}\sum_{j \in A} \log\big(1+u_j(\tilde{C}_{N,\gamma})\big) - \log2 - \frac{2\pi i}{Np}.
\end{equation}
By the same argument as used in the previous section, \eqref{eq: log cc2} transforms in the thermodynamic limit $N \rightarrow \infty$ to \eqref{eq: asym-cc}, indicating that the limit Cassini contours for $\lambda_1(\gamma)$ and $\lambda_2(\gamma)$ coincide
$$
\tilde{r}_{N,\gamma} := \big|\tilde{C}_{N,\gamma}\big|^{\frac{1}{p}} \rightarrow r_*(\gamma), \qquad N \rightarrow \infty.
$$
It means that as $N \rightarrow \infty$, $\tilde{r}_{N,\gamma} = (1+o(1))r_{N,\gamma}$. We prove \eqref{eq: asym-gap+} by estimating the asymptotic proximity of $\tilde{r}_{N,\gamma}$ to $r_{N,\gamma}$.

\begin{proof}[Proof of \eqref{eq: asym-gap+} in Theorem~\ref{thm: asym+}] \label{proof of asymLambda}
    Subtracting \eqref{eq: log cc1} from \eqref{eq: log cc2},
    \begin{equation} \label{eq: log cc2/cc1}
        \frac{1}{p}\sum_{j=1}^{p-1} \log(\frac{\tilde{Z}_j+1}{Z_j+1}) = \frac{1}{p}\log(\frac{Z_p+1}{\tilde{Z}_{p+1}+1}) + \frac{2 \pi i}{Np}.
    \end{equation}
    By Lemma \ref{lem: implicit}, $Z_j = u_j(C_{N,\gamma}) = v(w_j)$ for $j=1, \dots, p$ and similarly for $\tilde{Z}_j$. Hence, we denote
    $$
    \tilde{w}_j = (1+\xi)w_j,
    $$
    where $\xi=\xi_N=o(1)$ as $N \rightarrow \infty$ and is independent of $j$.
    Computing the derivatives of $v$ from the relation $(1-v)^p(1+v)^{N-p} = w^p$ and applying Taylor expansion to $v((1+\xi)w_j)$ in $\xi$, we have
    \begin{equation} \label{eq: serie Z_m}
        \begin{aligned}
            \tilde{Z}_j = Z_j &+ \frac{\xi \rho(Z_j^2-1)}{Z_j -1 + 2\rho} + o(|\xi|),  \qquad \text{for } j=1, \dots, p. 
        \end{aligned}
    \end{equation}
    Since $\tilde{r}_{N,\gamma} > r_\mathrm{cr}$ for $N$ sufficiently large, by expanding along the Cassini contour we obtain
    \begin{equation} \label{eq: Zp+1 serie}
        \tilde{Z}_{p+1} = \tilde{Z}_{p} + \frac{2 \pi i}{N} \frac{\tilde{Z}_p^2-1}{\tilde{Z}_p-1+2 \rho} + o \left( \frac{1}{N} \right).
    \end{equation}
    Inserting \eqref{eq: serie Z_m} and \eqref{eq: Zp+1 serie} into \eqref{eq: log cc2/cc1}, the right-hand reads
    $$
    \begin{aligned}
        \text{(RHS)} &= \frac{4\pi i \rho }{Np(Z_p-1+2\rho)} + o(|\xi|) + o \left( \frac{1}{N^2} \right)\\
        &= \frac{4\pi i \rho }{Np(u_*(r_*)-1+2\rho)} + o(|\xi|) + o \left( \frac{1}{N^2} \right),
    \end{aligned}
    $$
    where $u_*$ is determined by \eqref{eq: def-u*} and $r_* = r_*(\gamma)$ is the solution to \eqref{eq: asym-cc}, and we used the fact that $Z_p \rightarrow u_*(r_*)$ as $N\rightarrow\infty$.
    The left-hand side of \eqref{eq: log cc2/cc1} is equal to
    $$
    \text{(LHS)} = \frac{1}{p} \sum_{j=1}^{p-1} \frac{\xi\rho(Z_j-1)}{Z_j-1+2\rho} + o(|\xi|),
    $$
    where we apply the Lebesgue dominated convergence theorem to justify the interchange of the limit with the summation, since $C$ and $\tilde{C}$ are uniformly bounded, and $Z_j = u_j(C_{N,\gamma})$ and $\tilde{Z}_j = u_j(\tilde{C}_{N,\gamma})$ converge pointwise to some point of a limiting bounded curve parametrized by $j$.
    Using the argument in the proof of Proposition~\ref{prop: asym}, we can furthermore rewrite it as
    $$
    \begin{aligned}
        \text{(LHS)} &= \xi\rho \left[ \frac{1}{2\pi i \rho} \int_{\bar{u}_*(r_*)}^{u_*(r_*)} \frac{1}{u+1}\ du + o \left( \frac{1}{N} \right) \right] + o(|\xi|)\\
        &= \frac{\xi}{\pi} \arg\big(1+u_*(r_*)\big) + o(|\xi|) + o \left( \frac{1}{N^2} \right).
    \end{aligned}
    $$
    Hence, the leading order terms of \eqref{eq: log cc2/cc1} yield that
\begin{equation} \label{eq: xi found}
    \xi = \frac{i}{N^2} \cdot \frac{4\pi^2}{(u_*(r_*)-1+2\rho) \arg(1+u_*(r_*))} + o \left( \frac{1}{N^2} \right).
\end{equation}
    
    Finally, the spectral gap is presented as
    $$
    \begin{aligned}
        \lambda_1(\gamma) - \lambda_2(\gamma) &= \frac{1}{2}\sum_{j=1}^{p} \big(Z_j-\tilde{Z}_j\big) + \frac{1}{2}\big(\tilde{Z}_p - \tilde{Z}_{p+1}\big)\\
        &= \frac{1}{2}\sum_{j=1}^p \frac{\xi\rho(1-Z_j^2)}{Z_j-1+2\rho} + \frac{\pi i}{N} \frac{1-\tilde{Z}_p^2}{\tilde{Z}_p-1+2\rho} + o \left( \frac{1}{N} \right).
    \end{aligned}
    $$
    Applying the same argument as used in the proof of Proposition~\ref{prop: asym} to replace the summation by a contour integral,
    $$
    \lambda_1(\gamma) - \lambda_2(\gamma) = -\frac{N\xi\rho}{2\pi} \Im(u_*) + \frac{\pi i}{N} \frac{1-u_*^2}{u_*-1+2\rho} + o \left( \frac{1}{N} \right).
    $$
    The proof is then completed by substituting the formula \eqref{eq: xi found} of $\xi$. The expansion of the spectral gap \eqref{eq: asym-gap+} follows. 
\end{proof}

\begin{example} In the case of half-filling $N=2p$, see Examples \ref{half-fill} and \ref{ex:half-fill_r*}, we obtain 
\begin{equation} 
		\begin{split}
				 \lambda_1(\gamma) - \lambda_2(\gamma) = \frac{\pi}{N \sqrt{r_*(\gamma) -1}}\left(r_*(\gamma) - \frac{\sqrt{r_*(\gamma)-1}}{\arctan(\sqrt{r_* (\gamma)-1})}\right). 
			\end{split}
		\end{equation} 
\end{example}

\subsection{Asymptotic analysis for $\gamma<0$}

\subsubsection{The largest eigenvalue}
We employ the same analysis strategy: let set $A = \{1, \dots, p-1, N\}$ following Assumption \ref{ass: gammaneg} and the Proposition~\ref{prop: C}. This proposition indicates that the Cassini oval consists of two ovals; moreover, the positivity of $C_{N, \gamma}$ means $u_p(C_{N, \gamma})$ (the leftmost point on the right oval) and $u_N(C_{N, \gamma})$ (the rightmost point of the left oval) are real-valued, see Figure \ref{fig:  Cassini contour with sep curve} a).
The other choices in \eqref{choice2}, \eqref{eq:first-choice}, \eqref{eq:second-choice} do not yield a simple real eigenvalue. We linearize the consistency condition, define the relevant asymptotic contour, and extract the leading behaviour of the largest eigenvalue.

\begin{remark}
    When $\gamma < 0$, the Assumption \ref{ass: gammapos} regarding the choice of $A = \{1, \dots, p\}$, which works for $\gamma>0$, does not lead to a solution for large $N$. Indeed, the function $\mathcal{G}_{\infty}(u_*(r))$ on the right-hand side of the equation \eqref{eq: asym-cc} limits to zero for subcritical $r$ and takes finite positive values for supercritical $r$, which cannot match the negative finite value of $\gamma$ on the left-hand side. 
\end{remark}

Similarly, the consistency condition \eqref{eq: consistency} simplifies to the shift-invariance condition with zero momentum given in \eqref{def momentum}.
\begin{equation}
(u_{N}(C_{N,\gamma})+1) \prod_{m=1}^{p-1}  \Big( u_m(C_{N,\gamma}) +1\Big) = 2^p e^{\gamma p}.
\end{equation}
Taking the logarithm linearizes the equation.
\begin{equation}
\frac 1p \log (u_{N}(C_{N,\gamma})+1) + \frac 1p  \sum_{m=1}^{p-1} \log \Big( u_m(C_{N,\gamma}) +1\Big) = \log 2 + \gamma.
\end{equation}
Extending the sum to range from $1$ to $p$, we express the last equation in terms of the function $\mathcal{G}_N(C)$ defined in \eqref{def GN}.
\begin{equation} \label{eq: CCasympt neggamma}
\frac 1p \log \Big( \frac{u_N(C_{N,\gamma}) +1}{u_p(C_{N,\gamma})+1}\Big) +\mathcal{G}_N(C_{N,\gamma})  = \gamma.
\end{equation}
\begin{proposition} \label{prop: CN(gamma) asympt serie}
    For large $N$, the equation \eqref{eq: CCasympt neggamma} has a unique solution $C_{N, \gamma}^*$ that approaches zero exponentially fast, with the following leading asymptotic behaviour.
    \begin{equation} \label{eq: CN asymptotics}
    C_{N, \gamma}^* = e^{N^2 \gamma \rho(1-\rho)} 2^N   + o( e^{N^2 \gamma \rho(1-\rho)}2^N ).
\end{equation}
The largest eigenvalue asymptotic expansion follows.
\begin{equation} \label{la1negres}
    \lambda_1(\gamma) =  -1 + e^{N \gamma \rho}+e^{N \gamma (1-\rho)}+  o(e^{N \gamma \rho}+e^{N \gamma (1-\rho)}).
\end{equation}
\end{proposition}
\begin{proof}
For a non-zero, finite value of $ C_{N, \gamma}$, the first term of the left-hand side of the equation \eqref{eq: CCasympt neggamma} exhibits the following asymptotic behaviour as $N$ grows
\begin{equation}
    \frac{1}{p}\log \Big( \frac{u_N(C) +1}{u_p(C)+1}\Big) = \frac{1}{p} \begin{cases}
        \log (u_N(C) + 1) + O(1), & C < r_{\mathrm{cr}}^{1/\rho},\\
        O(1),& C \geq r_{\mathrm{cr}}^{1/\rho}.
    \end{cases}
\end{equation}
As $C$ decreases, this term diverges to negative infinity. 
For large values of $N$, the function $\mathcal{G}_N(C)$ estimate \eqref{eq:GNest_subcr} yields that $\mathcal{G}_N(C)$ is a subdominant term of order $O(|C|^{\frac{1}{p}}N^{-3})$ compared to $\log (u_N(C_{N,\gamma}) + 1)$ for any subcritical value of $C$ including $C=0$. 

 As a result, it does not contribute to the leading asymptotic behaviour. In contrast, for supercritical values of $C$,  $\mathcal{G}_N(C)$ provides the leading asymptotic behaviour of the left-hand side of equation \eqref{eq: CCasympt neggamma}. However, since \( \mathcal{G}_N(C) \) takes only positive real values, it cannot equal the negative values of $\gamma$.

The equation \eqref{eq: CCasympt neggamma}   has the following leading-order representation 
\begin{equation} \label{eq: CC_lead_order}
    \frac 1p \log \bigl(u_N(C_{N,\gamma}) + 1\bigr) + O\bigl((C_{N,\gamma})^{\frac 1p}N^{-3}\bigr) = \gamma + \frac 1p \log \bigl(u_p( C_{N,\gamma}) + 1\bigr).
\end{equation}
As $N$ increases, $u_N(C_{N,\gamma})$ should approach $-1$ monotonically with $C_{N, \gamma}$ approaching zero; therefore, there exists a unique asymptotic solution $C_{N, \gamma}^*$ limiting to zero as $N$ grows. From the defining equation of a Cassini oval $|1-u|^{\rho}|1+u|^{1-\rho}= C^{1/N}$, we derive the following series expansions at $u=1$ and $u=-1$, respectively.
\begin{equation} 
    u_p(C) = 1 - 2^{-\frac{1-\rho}{\rho}} C^{\frac 1p} + O(C^{\frac 2p}) , \qquad C \rightarrow 0,
\end{equation}
\begin{equation} 
    u_N(C) = -1 + 2^{-\frac{\rho}{1-\rho}}C^{\frac 1p \frac{\rho}{1-\rho}} + O\left(
    C^{\frac 2p \frac{\rho}{1-\rho}}
    \right) , \qquad C \rightarrow 0.
\end{equation}
By substituting the last two series into \eqref{eq: CC_lead_order}, we find $C_{N,\gamma}$ converges to zero exponentially fast, as described in \eqref{eq: CN asymptotics}. 
Consequently, we obtain the following results
\begin{equation} \label{eq:up asypt}
    u_p(C_{N,\gamma}^*) = 1 - 2 \ e^{N\gamma (1-\rho)} + O(e^{2N\gamma (1-\rho)}), \qquad  N\rightarrow +\infty.
\end{equation}
\begin{equation} \label{eq:uN asypt}
    u_N(C_{N,\gamma}^*) = -1 + 2 \ e^{N \gamma \rho} + O(e^{2N \gamma \rho}), \qquad N\rightarrow +\infty.
\end{equation}

 From an asymptotic relation \eqref{eq: asym-la} and the series expansions for $u_p$ \eqref{eq:up asypt} and for $u_N$ \eqref{eq:uN asypt}, we find the largest eigenvalue expansion 
\begin{equation}
\begin{split}
    \lambda_1(\gamma) = \frac12\sum_{j=1}^{p} \big(u_j(C_{N,\gamma}^*)-1\big) + \frac{1}{2} \big(u_N(C_{N,\gamma}^*)-u_p(C_{N,\gamma}^*) \big)
    \\
    = o(e^{N \gamma (1-\rho)}) + \frac{1}{2} \left(u_N(C_{N,\gamma}^*) -u_p(C_{N,\gamma}^*) \right)  \\
    = -1 + e^{N \gamma \rho}+e^{N \gamma (1-\rho)}+  o(e^{N \gamma \rho}+e^{N \gamma (1-\rho)}),
\end{split}
\end{equation}
where in the second row the sum was estimated by an integral and an error term, both of order $O(e^{N \gamma (1-\rho)})$ and cancelling each other.
\end{proof}
\begin{corollary}
The leading term of $\lambda_1(\gamma)$ is universal with respect to the parameter $\gamma$. 

Moreover, each choice with $p-1$ points from the right oval and one point from the left oval delivers an eigenvalue with the same leading asymptotic as the largest eigenvalue.
\end{corollary}

\subsubsection{The spectral gap}

Following the Assumption \ref{ass: gammaneg}, we consider the second-largest eigenvalue delivered by $A = \{1, \dots ,p-1,p+1\}$. The complex-conjugate eigenvalue comes from the choice $A = \{1, \dots ,p-1, N-1\}$. 

The difference between the largest eigenvalue  and the eigenvalue calculated with this selection is given by the equation
\begin{equation} \label{eq: spgap neg}
	2\left(\lambda_1(\gamma)-\lambda_2(\gamma) \right) = u_N(C_{N, \gamma})-u_{p+1}(\tilde{C}_{N, \gamma}) + \sum_{m=1}^{p-1} u_m(C_{N, \gamma})-\sum_{m=1}^{p-1} u_m(\tilde{C}_{N, \gamma}),
\end{equation}
where $\tilde{C}_{N, \gamma}$ is a parameter of the Cassini curve corresponding to $\lambda_2(\gamma)$.
The invariance condition \eqref{def momentum} with momentum $1$ simplifies to the following equation.
\begin{equation} \label{eq: CCasympt neggamma2}
\frac 1p \log \left( \frac{u_{p+1}(\tilde{C}_{N, \gamma}) +1}{u_p(\tilde{C}_{N, \gamma})+1}\right) +\mathcal{G}_N(\tilde{C}_{N, \gamma})  = \gamma + \frac{2 \pi i}{Np}.
\end{equation}

\begin{proposition}
For large $N$, the equation \eqref{eq: CCasympt neggamma2} has a unique solution $\tilde{C}_{N, \gamma}^*$  that approaches zero exponentially fast, following the leading asymptotic.
    \begin{equation} \label{eq: tildeCN asymptotics}
    \tilde{C}_{N, \gamma}^* = e^{N^2 \gamma \rho (1-\rho)} 2^{N} + o\left(e^{N^2 \gamma \rho (1-\rho)} 2^{N}\right).
    \end{equation}
The argument of the focal radius of a corresponding Cassini curve is 
\begin{equation} \label{eq: theta asymp}
\theta  = -\frac{2 \pi}{N} + o(N^{-1}).
\end{equation}
\end{proposition}
\begin{proof}
    Similar to the proof of the Proposition~\ref{prop: CN(gamma) asympt serie}, we demonstrate that for large $N$, the equation \eqref{eq: CCasympt neggamma2} has a unique solution, denoted as $\tilde{C}_{N, \gamma}^*$, which approaches zero as $N$ increases. Similarly, this conclusion follows from the series expansion in $\tilde{C}_{N, \gamma}^*$ limiting to zero as $N$ grows. From the defining polynomial equation $(1-u)^{\rho}(1+u)^{1-\rho}= |\tilde{C}|^{1/N} e^{i( \theta \rho+ \pi)}$ with $C$ small, we derive the following series expansions at $u=1$ and $u=-1$, respectively.
\begin{equation} 
    u_k(\tilde{C}) = 1 -2^{-\frac{1-\rho}{\rho}} |\tilde{C}|^{\frac 1p} e^{i \theta} e^{\frac{2 \pi i}{p}(k-p)}+ O(|\tilde{C}|^{\frac 2p}) , \quad  k = 1, \dots, p.
\end{equation}
\begin{equation} 
    u_{k}(\tilde{C}) = -1 + 2^{-\frac{\rho}{1-\rho}} \Big(|\tilde{C}|^{\frac 1p} e^{i\theta}\Big)^{\frac{\rho}{1-\rho}}
    e^{\frac{2 \pi i}{N-p}(k-p)} + O\left(|\tilde{C}|^{\frac 2p \frac{\rho}{1-\rho}}\right) , \ k = p+1, \dots, N
\end{equation}
By substituting these expansions into equation \eqref{eq: CCasympt neggamma2}, and taking into account that we defined $\theta  := p^{-1} \arg(-C) = o(1)$, we derive the following asymptotic equation
\begin{equation} 
\begin{split}
\Big(\frac{|\tilde{C}|^{\frac 1p} e^{i\theta}}{2}\Big)^{\frac{\rho}{1-\rho}}
    e^{\frac{2 \pi i}{N-p}} &+ O\left(|\tilde{C}|^{\frac 2p \frac{\rho}{1-\rho}} \right) \\
    &= e^{p\gamma + \frac{2 \pi i}{N} + O(N^{-2}|\tilde{C}|^{\frac 1p})} \Big(2 -  \frac{|\tilde{C}|^{\frac 1p} e^{i \theta} }{2^{-\frac{1-\rho}{\rho}}}e^{\frac{2 \pi i}{p}(k-p)}+ o(|\tilde{C}|^{\frac 1p})\Big).
    \end{split}
\end{equation}
The leading asymptotic of $C_{N,\gamma}^*$ has the same absolute value as that obtained for the largest eigenvalue as stated in Proposition~\ref{prop: CN(gamma) asympt serie}. For the argument parameter $\theta$, we find \eqref{eq: theta asymp}.
\end{proof}
From this proposition, we obtain the expansion series in large $N$ 
\begin{equation} 
    u_k(C_{N,\gamma}^*) = 1 - 2 e^{N \gamma (1-\rho)} e^{-\frac{2 \pi i}{N}+\frac{2 \pi i}{p}(k-p)}+ O(e^{2N \gamma (1-\rho)}), \quad  k = 1, \dots, p,
\end{equation}
\begin{equation} 
    u_k(C_{N,\gamma}^*) = -1 + 2 e^{N \gamma \rho} e^{-\frac{2 \pi i}{N}\frac{\rho}{1-\rho}}
    e^{\frac{2 \pi i}{N-p}(k-p)}+ O\left(e^{2N \gamma \rho}\right), \quad k = p+1, \dots, N.
\end{equation}

For the spectral gap \eqref{eq: spgap neg}, we have two exponentially small contributions: the first one comes from $u_N(C_{N, \gamma})-u_{p+1}(\tilde{C}_{N, \gamma})$ and the second one from the difference in the sums $\sum_{j=1}^{p-1} (u_j(C_{N, \gamma})- u_{j}(\tilde{C}_{N, \gamma})) = u_p(C_{N, \gamma})-u_{p}(\tilde{C}_{N, \gamma})$.
\begin{equation}
\begin{split}
    \lambda_1(\gamma) - \lambda_2(\gamma) &= e^{N \gamma \rho} \biggl(1 - e^{\frac{2 \pi i}{N}}\biggr) +  e^{N \gamma (1-\rho)} \biggl(1 - e^{-\frac{2 \pi i}{N}}\biggr) \\
    &\quad + O\bigl(e^{2N \gamma \rho} + e^{2N \gamma (1-\rho)}\bigr) \\
    &= \frac{2\pi i}{N} \Bigl(- e^{N \gamma \rho} + e^{N \gamma (1-\rho)} \Bigr) + \frac{4\pi^2}{N^2}  \Bigl(e^{N \gamma \rho} + e^{N \gamma (1-\rho)}\Bigr) \\
    &\quad + O\bigl(e^{2N \gamma \rho} + e^{2N \gamma (1-\rho)}\bigr).
\end{split}
\end{equation}
\begin{remark}
    If we choose a point other than $Z_{p+1}$ from $Z_{p+2}, \dots, Z_{N-1}$ in the left oval, the leading asymptotic behaviour of the focal radius is still described by \eqref{eq: CN asymptotics}.
\end{remark}
\newpage
\appendix 

\newpage
\section{Cassini oval boundedness }
\label{app boundedness}
\begin{proposition}  \label{prop:C bounded}
\begin{itemize}
    \item [(i)] For any $N,p$ and a non-zero finite value of parameter $\gamma$, the parameter $r_N(\gamma) = |C_{N, \gamma}|^{1/p}$, defining Cassini ovals \eqref{eq: Cassini contours} is uniformly bounded from above 
    $$r_N(\gamma) < \left(2e^{\gamma}+2\right)^{\frac{1}{\rho}}.$$
    \item [(ii)] For the solution delivering the largest and next to largest eigenvalue the corresponding constant $C$ is also uniformly bounded from zero. 
\end{itemize}

\end{proposition}
\begin{proof}  (i) Let $N,p, \gamma$ non-zero fixed, $Z_m = u_{l(m)}(r_N)$ be a solution to Bethe equations for some choice function $l$. From the shift-invariance condition \eqref{def momentum}, to which all the solutions $u_{l(m)}(r)$ satisfy, we  obtain the following relation on absolute values.
\begin{equation} \label{eq: CC 102}
    \prod_{m=1}^{p} |Z_m +1| = 2^p e^{\gamma p}.
\end{equation}
These values also satisfy Cassini oval equation given by \eqref{eq: Cassini contours} becoming 
\begin{equation}
	  	|1-Z_m|^{\rho}|1+Z_m|^{1-\rho}= r_{N}^{\rho}.
\end{equation}
Parametrizing $1+Z_m = a e^{i \psi}$ and using the Cassini curve equation, we have
\begin{equation} \label{eq: Cassini in R}
	  	(a^2 - 4a \cos \psi + 4) a^{\frac{2(1-\rho)}{\rho}}= r_{N}^{2}.
\end{equation}
Estimating the left-hand side of this equation from above by $(a+2)^{\frac{2}{\rho}}$, we see that $a > r_{N}^{\rho}-2 > 0$ for sufficiently large $r_N$. Combining with \eqref{eq: cc1 simple} we have the following bound 
\begin{equation}
     r_{N}^{\rho}-2 <  \sqrt[p]{\prod_{m=1}^{p} |Z_m +1|}= 2 e^{\gamma}
\end{equation}
Therefore, the bound from above is $r_{N}< \left( 2e^{\gamma}+2\right)^{\frac{1}{ \rho}}$. 

\textit{(ii)} Assume that $r_{N}$ corresponding to the largest eigenvalue is close to zero. Then, the right Cassini oval can be approximated by a circle of radius $2^{\frac{\rho-1}{\rho}}r_{N}$ around $1$, so that for each $Z_m$ from this oval we have 
\begin{equation}
    2-2^{\frac{\rho-1}{\rho}}r_{N} < |1+Z_m| < 2+2^{\frac{\rho-1}{\rho}}r_{N}.
\end{equation}
Since for the largest eigenvalue all $Z_m$ are taken from the right oval, the equation \eqref{eq: CC 102} yields the bounds $r_{N}>2 |e^{\gamma}-1|$ for a non-zero $\gamma$. 
For some $\gamma$ this bound contradicts the assumption that $r_N$ is close to zero (or even $r_{N}<r_{\mathrm{cr}}$), but for finite $\gamma$ close to $0$ it provides a uniform bound in terms of $\gamma$. The similar bound for the next to the largest eigenvalue follows for large $p$ as the contribution of one $Z_m$ taken from the left oval while the remaining $p-1$ belong to the right one becomes negligible in the limit of large $N$.
\end{proof}

\newpage
\section{Contour integration}   \label{app: contour integration}
We fix $N,p$. 
In this part, we explain in detail how to estimate a complex integral 
\begin{equation} \label{eq: I(f,C)}
I_N(f,r)  =  \frac{1}{p} \int_{1}^{p} f \Big( v(w(z))\Big) \ dz
\end{equation}	
for a meromorphic complex-valued function $f(\cdot)$ defined on a domain $D_{+}$. The function $w(z)$ parametrizes a circle of radius $r$
\begin{equation}
    w(z):= r e^{\frac{2 \pi i}{p}(z-\frac{1}{2})}
\end{equation}
and $v(\cdot)$ is the inverse function of $w(z)$. 

We recall that the inverse function $v(\cdot)$ is not defined on $\partial D_+$. Therefore, to avoid $u_p$ lying on the boundary, i.e. the separating curve, we formulate the following lemmas for a real $r < r_{\mathrm{cr}}$.

\begin{figure}
\center{\includegraphics[width=1\linewidth]{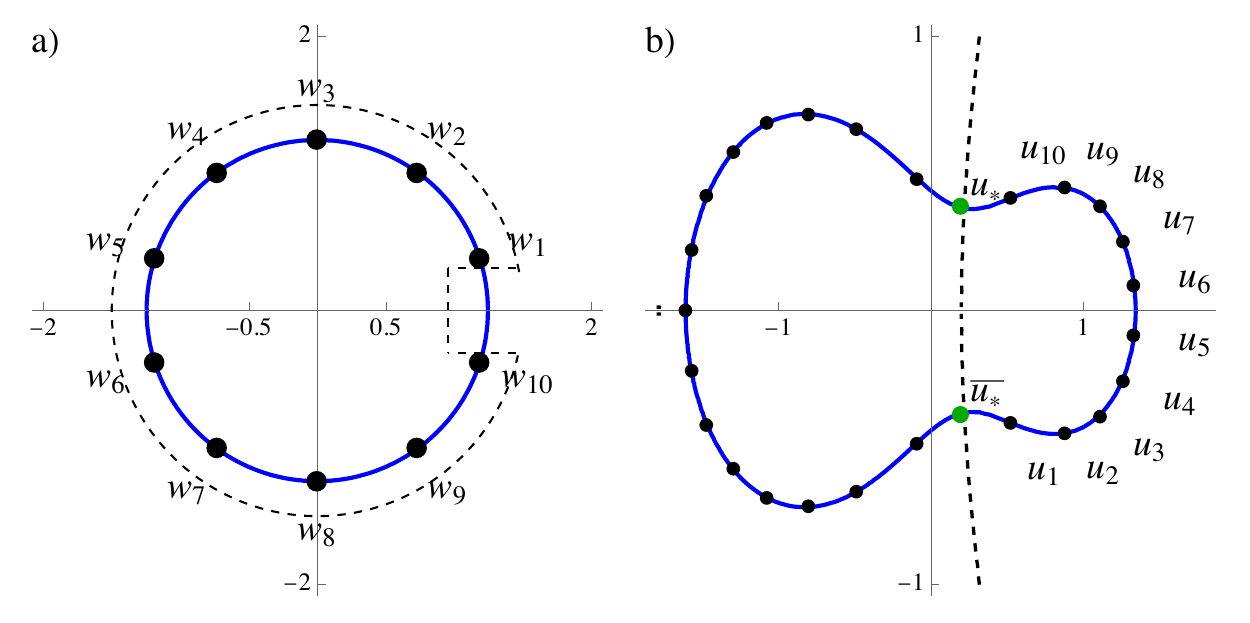}}
\caption{a) shows the original circular contour on a complex plane with a branch cut $[r_{\mathrm{cr}}, +\infty]$ along a real axis. 
The points $w_1$ to $w_p$ are shown for a real $C_{N,\gamma} < -r_{\mathrm{cr}}$. b) displays the equivalent Cassini oval contour $\Gamma(r)$ located in a domain $D_+$ on the right of a separating curve drawn by a dashed line.}
\label{fig: contour integration}
	\end{figure}
    
\begin{lemma} \label{lemB1: change of variables}
The integral $I_N(f,r)$ in terms of a variable $u=v(w(z))$ is given by
\begin{equation} \label{eq: I_N over Cassini}
I_N(f,r) = \frac{1}{2 \pi i \rho}\int_{\Gamma(r)}  \frac{(u-1+2 \rho) }{u^2-1} f(u) \ du,
\end{equation}
where the contour of integration $\Gamma(r)$ is a counter-clockwise Cassini contour segment from $u_1$ to $u_p$, as illustrated in Figure \ref{fig: contour integration}b).
\end{lemma}
\begin{proof}
The parametrization of the integral \eqref{eq: I(f,C)} suggests a change of variables: instead of integrating over the interval from $1$ to $p$, we can integrate over a circle of radius $r$ walking counter-clockwise from $w_1$ to $w_p$ with a branch cut $[1, +\infty)$ along the real axis on a complex plane, see Figure~\ref{fig: contour integration}a). Alternatively, the integration can be performed along a Cassini oval starting from $u_1$ to $u_p$ in a counter-clockwise direction, see Figure~\ref{fig: contour integration}(b). 
Using the inverse function rule, we obtain the following change in differentials
\begin{align*}
    du &= \frac{dv}{dw} \cdot \frac{dw}{dz} \ dz \\
       &= \frac{1}{w'(u)} \cdot \frac{dw}{\ dz} \ dz \quad \text{(by inverse function derivative in $D_+^{\circ}$)} \\
       &= \frac{\rho(u^2 - 1)}{u - 1 + 2\rho} \cdot \frac{2\pi i}{p} \ dz \quad \text{(from \eqref{def: w(u)})}.
\end{align*}
Substituting into \eqref{eq: I(f,C)} yields \eqref{eq: I_N over Cassini}. 
\end{proof}

\begin{lemma} \label{lemB2:contour integration}
Let $\Gamma'(r) = [u_1,0] \cup [0, u_p]$. Assume a function $f(\cdot)$ has no poles on $\Gamma'(r)$. Closing the contour of integration $\Gamma(r)$ by a path from $u_1$ to $u_p$ through $0$, so that $\Gamma(r) =  \Gamma'(r) + \Gamma''(r)$, the integral $I_N(f,r)$ splits into two terms.
\begin{equation} \label{eq: I = res+int}
\begin{split}
I_N(f,r) =
		\sum_{q \in Q }  \res_{u = q}
		&\left( \frac{(u-1+2 \rho) }{\rho (u^2-1)} f(u) \right) + \frac{1}{2 \pi i} \int_{\Gamma'(r)}  \frac{(u-1+2 \rho) }{\rho (u^2-1)} f(u) \ du.
\end{split}
\end{equation}
The first term is a sum over all residues at points in the set $Q$ of an integrand function inside the closed contour $\Gamma''(r)$, and the second integral goes over $\Gamma'(r)$.

For a real $C_{N, \gamma}< 0$, the points $u_1$ and $u_p$ become complex conjugate while the last integral of an odd function $f(\cdot)$ over a symmetric curve is real-valued. In particular, for $f(z) = z-1$ and $f(z) = \log (z+1)$, we have 

\begin{equation} \label{eq: int rep of Lambda}
  I_N(z-1, r) =
   \frac{1}{\pi \rho} \Big( \Im u_p + 2(\rho-1) \arg (1+u_p)\Big),
\end{equation}
\begin{equation}
\begin{split}
  I_N\Big(\log (z+1), r \Big) = \log 2 + \frac{1}{2\pi} &\Big(\log 2 \arg\left(\frac{u_p-1}{u_1-1}\right) -2 \Im \Big( \operatorname{Li}_2 \Big(\frac{1-u_p}{2}\Big)\Big) \\
  &+ \frac{1-\rho}{\rho} \log|1+u_p| \arg \left(u_p+1\right)\Big).
  \end{split}
\end{equation}
\end{lemma}

\begin{proof}
Indeed, the first integral over the closed contour can be evaluated using the Cauchy residue theorem, and the second integral along two segments from $u_1$ to $u_p$ remains. 

\textbf{Case $f(z) = z - 1.$} The integrand of $I_N(z-1, r)$ has no residues inside the
closed contour $\Gamma''(r)$; therefore, only the second term remains  
\begin{equation} 
\begin{split}
   I_N(z-1, r) &= \frac{1}{2 \pi i \rho} \int_{\Gamma'(r)}    \frac{u-1+2 \rho }{ u+1}  \ du\\
	&=\frac{1}{2 \pi \rho} \Im \left(u_p-u_1 + 2(\rho-1)\log \frac{1+u_p}{1+u_1} \right). 
    \end{split}
\end{equation}
For a complex conjugate $u_1$ and $u_p$ the further simplification gives \eqref{eq: int rep of Lambda}.

\textbf{Case $f(z) = \log(z+1).$} The closed contour $\Gamma''(r)$ includes a unique pole at $u=1$, yielding the residue $\log 2$. 
\begin{equation} 
\begin{split}
		I_N\Big(\log (z+1), r \Big) = 
			\log 2 
            + \frac{1}{2 \pi i \rho} \int_{\Gamma'(r) }  \frac{(u -1 + 2 \rho) }{u^2-1} \log (1+u) \ du 
        \end{split}
\end{equation}
For a complex conjugate $u_1$ and $u_p$ only the contribution from the imaginary part remains 
\begin{equation} \label{eq: G(Zp) int}
\begin{split}
  I_N\Big(\log (z+1), r \Big) = \log 2
  + \Im \Big[\frac{1}{ 2 \pi \rho} \int_{u_1}^{u_p}  \frac{(u -1 + 2 \rho) }{u^2-1} \log (1+u) \ du \Big] 
  \end{split}
\end{equation}
The last integral could be split into terms 
\begin{equation}
\begin{split}
   \int \frac{(u -1 + 2 \rho) }{u^2-1} \log (1+u)\ du  = \int \left(\frac{\rho}{u-1} + \frac{1-\rho}{u+1} \right) \log(1+u) \ du\\
   =\rho \left( \log 2 \log(u-1) + \int \frac{\log(1+\frac{u-1}{2})}{u-1} \ du \right)+ \frac 12(1-\rho) (\log(1+u))^2,
\end{split}
\end{equation}
and integrated in terms of the dilogarithm $\operatorname{Li}_2(z) := -\int_0^z \frac{\ln(1-t)}{t}dt$ with $z \notin [1, +\infty]$ appearing naturally from the middle term. 
Indeed, in a new variable $\tilde{u} = -\frac{u-1}{2}$, we have
\begin{equation}
\begin{split}
    \int_{\Gamma'(r)} \frac{\log(1+\frac{u-1}{2})}{u-1} \ du  &= \int_{\frac{1-u_1}{2}}^{\frac{1}{2}} \frac{\log(1-\tilde{u})}{\tilde{u}} d\tilde{u} + \int_{\frac{1}{2}}^{\frac{1-u_p}{2}} \frac{\log(1-\tilde{u})}{\tilde{u}} d\tilde{u}\\
    &= \operatorname{Li}_2\left(\frac{1-u_1}{2}\right)-\operatorname{Li}_2\left(\frac{1-u_p}{2}\right)
\end{split}
	\end{equation}
Finally, we collect all the terms  
\begin{equation}
\begin{split}
  I_N\Big(\log (z+1), r \Big) &=  \frac{1-\rho}{\pi \rho} \log|1+u_p| \arg \left(u_p+1\right)\\
  &= \log 2 + \frac{1}{\pi} \Big(\log 2 \arg(1-u_p) - \Im \Big( \operatorname{Li}_2 \Big(\frac{1-u_p}{2}\Big)\Big)\Big) \\
  &+ \frac{1-\rho}{\pi \rho} \log|1+u_p| \arg \left(u_p+1\right).
  \end{split}
\end{equation}

\end{proof}

\begin{lemma} \label{lemB3: asymptotis of IN}
    For a real supercritical $C_{N, \gamma}< -r_{\mathrm{cr}}^p$, the integrals $I_N(z-1, r )$ and $I_N(\log (z+1), r )$ have the following asymptotic expansion as $N$ and $p$ grow, keeping $\rho$ fixed.
\begin{equation} 
  I_N(z-1, r) =
   \frac{1}{\pi \rho} \Big( \Im u_* + 2(\rho-1) \arg (1+u_*)\Big) +O(N^{-1}). 
\end{equation}
\begin{equation}
\begin{split}
  I_N\Big(\log (z+1), r \Big) = \log 2 + \frac{1}{2\pi} &\Big( \log 2 \arg\left(\frac{u_*-1}{\overline{u_*}-1}\right) - 2\Im \Big( \operatorname{Li}_2 \Big(\frac{1-u_*}{2}\Big)\Big) \\
  &+ \frac{1-\rho}{\rho} \log|1+u_*| \arg \left(u_p+1\right)\Big) +O(N^{-1}),
  \end{split}
\end{equation}
\end{lemma}
\begin{proof}
As $N,p \to +\infty$ with fixed $\rho$ and  $r \neq r_{\mathrm{cr}}$, the discrete points $\{u_j\}$ become uniformly dense on a Cassini oval. Taylor-expanding $u_{p+1}$ around $u_p$, we observe
\begin{equation} 
    u_{p+1}(r) = u_{p}(r) + \frac{2 \pi i}{N} \frac{u_p^2(r)-1}{u_p(r)-1+2 \rho} + o(N^{-1}).
\end{equation}
with the characteristic distance between the neighbouring points of order $O(N^{-1}).$
Thus, $u_p \to u_*$ with an error $\mathcal{O}(N^{-1})$. Substituting $u_*$ into previous results gives the asymptotic forms.

\end{proof}
\begin{example}
    At half-density the function \eqref{eq: G(Zp) int} takes the following form
	\begin{equation} 
		I_N(z-1, r)  =
			\log 2 + \frac{2}{ \pi } \int_{0}^{ \arctan(\sqrt{r-1})}  x \ \tan(x) dx, \quad \rho = \frac 12, 
	\end{equation}    
where we used the equation \eqref{half-fillr*} and changed a variable $u = i \tan x$.
\end{example}

\section*{Acknowledgements}
We thank Stefano Olla, Ofer Zeitouni, and Alexander Povolotsky for insightful discussions and valuable comments that greatly contributed to this work. We are also grateful to Gunter Schütz for generously sharing important references. Additionally, we thank Artem Lipin for his assistance with the code used to generate Figures \ref{fig: eigvals pos gamma} and \ref{fig: eigvals neg gamma}.

 \printbibliography

\end{document}